\documentclass{article}

\usepackage{arxiv}

\usepackage{wrapfig}

\usepackage[colorlinks=true,urlcolor=purple,
linkcolor=purple,citecolor=purple]{hyperref}

\AtBeginDocument{
  \label{CorrectFirstPageLabel}
  
}

\usepackage{graphicx,caption}
\usepackage{mathrsfs}
\usepackage[dvipsnames]{xcolor}

\usepackage{multirow}
\usepackage{nicefrac} 
\usepackage{bbm}
\usepackage{array}
\usepackage{float}
\usepackage{color}

\usepackage{amsthm}
\usepackage{mathtools}
\usepackage{amssymb}
\usepackage{amsmath}

\usepackage{url}

\usepackage{algorithm}
\usepackage{algpseudocode}
\usepackage{thmtools,thm-restate} 

\makeatletter
\def\paragraph{\@startsection{paragraph}{4}%
  \z@\z@{-\fontdimen2\font}%
  {\normalfont\bfseries}}
\makeatother

\usepackage{tikz}
\usetikzlibrary{decorations.markings, backgrounds, calc, positioning, shapes, shadows, arrows, fit, automata}
\usepackage{tikz-cd}
\usetikzlibrary{arrows.meta}
\tikzset{>={Latex}}

\newcommand{\N}{\mathbb{N}}

\newcommand{\R}{\mathbb{R}}

\newcommand{\mc}{\mathcal}

\newcommand{\mf}{\mathfrak}

\renewcommand{\a}{\alpha}
\renewcommand{\b}{\beta}
\newcommand{\g}{\gamma}
\renewcommand{\d}{\delta}
\newcommand{\e}{\varepsilon}
\newcommand{\w}{\omega}
\newcommand{\s}{\sigma}
\newcommand{\ph}{\varphi}

\renewcommand{\t}{\tau}

\newcounter{desccount}

\newcommand{\descref}[1]{\hyperref[#1]{#1}}

\newcommand{\leb}{\mathscr{L}}

\DeclareMathOperator*{\esssup}{ess\,sup}

\newcommand{\lp}{\left(}
\newcommand{\rp}{\right)}
\newcommand{\lbar}{\left|}
\newcommand{\rbar}{\right|}
\newcommand{\lnorm}{\left\|}
\newcommand{\rnorm}{\right\|}

\newcommand\mattwo[4]{\left(\begin{smallmatrix}
			{#1} & {#2}\\
     			{#3} & {#4}
                     \end{smallmatrix}\right)}
\newcommand\mattres[9]{\left(\begin{smallmatrix}
			{#1} & {#2}&{#3}\\
     			{#4} & {#5}&{#6}\\
     				{#7} &{#8} &{#9}
                     \end{smallmatrix}\right)}

\newcommand\vecthree[3]{\left(\begin{smallmatrix}
			{#1} \\
     			{#2} \\
			{#3}
                     \end{smallmatrix}\right)}

\newcommand\vecfour[4]{\left(\begin{smallmatrix}
			{#1} \\
     			{#2} \\
			{#3} \\
			{#4}
                     \end{smallmatrix}\right)}

\newcommand{\borel}{\operatorname{Borel}}	

\newcommand{\set}[1]{\left\{#1\right\}}
\newcommand{\la}{\langle}
\newcommand{\ra}{\rangle}

\renewcommand{\r}{\rightarrow}

\newcommand{\norm}[1]{\|#1\|} 

\newcommand{\supp}{\operatorname{supp}}

\newcommand{\inv}{^{-1}}

\newcommand{\diff}{\, d}

\newcommand{\myint}{\int\!}

\newcommand{\Nm}{\Ngen}

\newcommand{\Ngen}{\mathcal{N}}

\newcommand{\dis}{\operatorname{dis}}		

\newcommand{\prob}{\operatorname{Prob}}

\newcommand{\dn}{d_{\mathcal{N}}}	
\newcommand{\dnp}{d_{\mathcal{N} \!,p}} 
\newcommand{\dgp}{d_{\mathcal{N} \!,\a}^{\, \mc{G}\mc{P}}} 
\newcommand{\dgpzero}{d_{\mathcal{N} \!,0}^{\, \mc{G}\mc{P}}} 
\newcommand{\dnt}{d_{\mathcal{N} \!,2}} 
\newcommand{\dninf}{d_{\mathcal{N} \!,\infty}}

\newcommand{\coup}{\mathscr{C}}

\newcommand{\di}{d_{\operatorname{I}}}

\newcommand{\sph}{\mathbb{S}}

\newcommand{\1}{\mathbf{1}}
\newcommand{\size}{\operatorname{size}}
\newcommand{\diag}{\operatorname{diag}}

\newcommand{\subw}{\operatorname{subSize}}
\newcommand{\supw}{\operatorname{supSize}}

\newcommand{\eccout}{\operatorname{ecc}^{\operatorname{out}}}
\newcommand{\eccin}{\operatorname{ecc}^{\operatorname{in}}}

\newtheorem{theorem}{Theorem}

\newtheorem{proposition}[theorem]{Proposition}
\newtheorem{lemma}[theorem]{Lemma}

\theoremstyle{definition}
\newtheorem{example}[theorem]{Example}
\newtheorem{remark}[theorem]{Remark}

\newtheorem{definition}{Definition}

\makeatletter
\newcommand{\pushright}[1]{\ifmeasuring@#1\else\omit\hfill$\displaystyle#1$\fi\ignorespaces}
\newcommand{\pushleft}[1]{\ifmeasuring@#1\else\omit$\displaystyle#1$\hfill\fi\ignorespaces}
\makeatother

\title{The Gromov-Wasserstein distance between networks and stable network invariants}

\author{
  Samir Chowdhury\\
  Department of Mathematics\\
  The Ohio State University\\
  Columbus, Ohio 43210 \\
  \texttt{chowdhury.57@osu.edu} \\
   \And
 Facundo M\'{e}moli \\
  Department of Mathematics\\ 
  Department of Computer Science and Engineering\\
  The Ohio State University\\
  Columbus, Ohio 43210 \\  
  \texttt{memoli@math.osu.edu} \\
}

\begin{document}
\maketitle

\newif\ifdone 

\begin{abstract} 
We define a metric---the network Gromov-Wasserstein distance---on weighted, directed networks that is sensitive to the presence of outliers. In addition to proving its theoretical properties, we supply network invariants based on optimal transport that approximate this distance by means of lower bounds. We test these methods on a range of simulated network datasets and on a dataset of real-world global bilateral migration. For our simulations, we define a network generative model based on the stochastic block model. This may be of independent interest for benchmarking purposes.

\end{abstract}


\setcounter{tocdepth}{3}
\tableofcontents

\section{Introduction}

\subsection{Motivation and related literature}

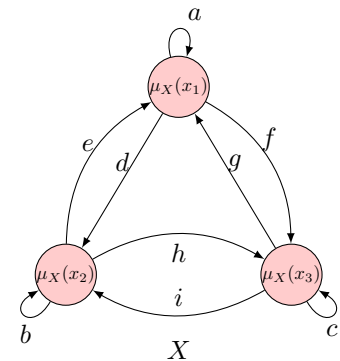
\begin{wrapfigure}{r}{0.3\textwidth} 
\begin{center}
\begin{tikzpicture}
\tikzset{>=latex}
\node[circle,draw, inner sep=0pt, scale=0.7, fill=red!20](1) at (0,2.5){$\mu_X(x_1)$};
\node[circle,draw, inner sep=0pt, scale=0.7, fill=red!20](2) at (-1.5,0){$\mu_X(x_2)$};
\node[circle,draw, inner sep=0pt, scale=0.7, fill=red!20](3) at (1.5,0){$\mu_X(x_3)$};

\path[->] (1) edge [loop above, min distance = 5mm] node[above right]{$a$}(1);
\path[->] (2) edge [loop, out=240, in = 210, min distance = 5mm] node[below]{$b$}(2);
\path[->] (3) edge [loop, out=300,in=330, min distance = 5mm] node[below]{$c$}(3);
\path[->] (1) edge [bend left,in=180,out=0] node[above,pos=0.5]{$d$} (2);
\path[->] (2) edge [bend left] node[above,pos=0.5]{$e$} (1);
\path[->] (1) edge [bend left] node[above,pos=0.5]{$f$} (3);
\path[->] (3) edge [bend left,out=0,in=180] node[above,pos=0.5]{$g$} (1);
\path[->] (2) edge [ bend left] node[below,pos=0.5]{$h$} (3);
\path[->] (3) edge [bend left] node[above,pos=0.5]{$i$} (2);

\node at (0,-1){$X$};
\end{tikzpicture}
\end{center}
\captionsetup{width=\linewidth}
\caption{The networks in this work have asymmetric pairwise weights and a significance value for each node.}
\label{fig:nets-header}
\end{wrapfigure}

Advances in data mining are beginning to lead to the acquisition of large networks that are directed, weighted, and possibly even signed \cite{newman2010networks}. In light of the ready availability of such data, a natural problem is to devise methods for comparing network datasets. These methods in turn lead to a wide range of applications. An example is the \emph{network retrieval task}: given a database of networks and a query network, return an ordered list of the networks in the database that are most similar to the query. Additionally, because there may be redundant data in the networks that are not relevant to the query, one may wish to impose a notion of significance to certain substructures of the query network. The task then is to retrieve networks which are similar to the query network both globally and also at the scale of relevant substructures.

While there has been some work in devising directed, weighted analogues of conventional network analysis tools such as edge overlap and clustering coefficients, we are more interested in pairwise comparison of individual networks. The intuitive idea behind this comparison is to search for the best possible alignment of edges (according to weights) while simultaneously aligning nodes with similar significance. 

Techniques based on optimal transport (OT) provide an elegant solution to this problem by endowing a network with a probability measure. The user adjusts the measure to signify important network substructures and to smooth out the effect of outliers. This approach was adopted in \cite{hendrikson} to compare various real-world network datasets modeled as \emph{metric measure (mm) spaces}---metric spaces equipped with a probability measure. This work was based in turn on the formulation of the \emph{Gromov-Wasserstein (GW) distance} between mm spaces presented in \cite{dgh-sm, dghlp}. Specifically, this setting considered triples $(X,d_X,\mu_X)$ where $(X,d_X)$ is a compact metric space and $\mu_X$ is a Borel probability measure.

Exact computation of GW distances amounts to solving a nonconvex quadratic program. Towards this end, the computational techniques presented in \cite{dgh-sm, dghlp} included both readily-computable lower bounds and an alternate minimization scheme for reaching a local minimum of the GW objection function. This alternate minimization scheme involved solving successive linear optimization problems, and was 
used for the computations in \cite{hendrikson}.

An alternative definition of the GW distance due to Sturm (the \emph{transportation} formulation) appeared in \cite{sturm2006geometry}, although this formulation is less amenable to practical computations than the one in \cite{dgh-sm} (the \emph{distortion} formulation). Both the transportation and distortion formulations were studied carefully in \cite{dgh-sm, dghlp,sturm2012space}. 
It was further observed by Sturm in \cite{sturm2012space} that the definition of the (distortion) GW distance can be extended to \emph{gauged measure spaces} of the form $(X,\hat{d}_X,\mu_X)$. Here $X$ is a Polish space, $\hat{d}_X$ is a symmetric $L^2$ function on $X\times X$ (that does not necessarily satisfy the triangle inequality), and $\mu_X$ is a Borel probability measure on $X$. These results are particularly important in the context of the current paper. From here on, we always refer to the distortion formulation of the GW distance.

Sturm's work in \cite{sturm2012space} showed that while the collection of isomorphism classes of metric measure spaces is not complete, elements in its completion can be represented by triples $(X,\tilde{d}_X,\mu_X)$ where $X,\mu_X$ are as above, and $\tilde{d}_X$ is a symmetric, measurable, square integrable function satisfying the triangle inequality almost everywhere. He further showed that the ambient space of gauged measure spaces, which is interpreted as being ``more linear" due to giving up the triangle inequality, admits explicit descriptions of geometric properties. 

In Sturm's work, symmetry is desirable because, for example, it allows an easy definition of open balls, whose volume growth is of theoretical interest (the asymmetric case would require ``forward-open" and ``backward-open" balls). However, practical data is often characterized by lack of symmetry, e.g. inhibitory/excitatory effects in neurons, unidirectional gene regulation in cell signaling pathways, and human migration between countries. The asymmetric case is of primary interest in the current work.  

From now on, we reserve the term \emph{network} for network datasets that cannot necessarily be represented as metric spaces, unless qualified otherwise. An illustration is provided in Figure \ref{fig:nets-header}. Already in \cite{hendrikson}, it was observed that numerical computation of GW distances worked well for comparing graph-structured data even when the underlying datasets failed to be metric. This observation was further developed in \cite{pcs16}, where the focus from the outset was to compute the GW distance (and related discrepancies) between arbitrary matrices, i.e. what we refer to as finite networks. While the experiments of \cite{pcs16} were on symmetric datasets, their implementations remain valid and theoretically justified even on matrices that do not satisfy symmetry. We emphasize this point in the current work, and extend from matrices to the continuous setting. Thus this work should be viewed as a theoretical complement to \cite{pcs16}.

On the computational front, the authors of \cite{pcs16} directly attacked the nonconvex optimization problem by considering an \emph{entropy-regularized} form of the GW distance (ERGW) following \cite{s16}, and using a projected gradient descent algorithm based on results in \cite{benamou2015iterative,s16}. This approach was also used (for a generalized GW distance) on graph-structured datasets in \cite{vayer2019optimal}. It was pointed out in \cite{vayer2019optimal} that the ERGW approach occasionally requires a large amount of regularization to obtain convergence, and that this could possibly lead to over-regularized solutions. A different approach, developed in \cite{dgh-sm,dghlp}, considers the use of lower bounds on the GW distance as opposed to solving the full GW optimization problem. This is a practical approach for many use cases, in which it may be sufficient to simply obtain lower bounds for the GW distance. One of the lower bounds in \cite{dgh-sm} involved linearizing the GW objective by decoupling the alignment term into two separate terms (thus removing the quadratic dependence), and optimizing over each term separately (referred to as the \emph{Third Lower Bound} (\ref{inv:TLB})). This approach was also used in \cite{schmitzer2013modelling}, with a further relaxation of one of the marginal terms. 

As a complement to the alternate minimization scheme of \cite{dghlp} and the ERGW scheme of \cite{pcs16}, our numerical experiments are carried out using the lower bound approach, specifically the (\ref{inv:TLB}). This is certainly faster than alternate minimization (see \cite{hendrikson} for computational aspects), but potentially slower than the ERGW scheme of \cite{pcs16}. However, it has the benefit of not needing any parameter tuning, which is an issue with entropic regularization. This makes it useful for exploratory network data analysis.

\subsection{Contributions} We adopt the setting of networks $(X,\w_X,\mu_X)$, where $X$ is a Polish space, $\mu_X$ is a Borel probability measure, and $\w_X$ is any measurable, integrable function on $X\times X$ (decoupled from the topology of $X$ beyond Borel measurability). Using the GW distance formulation, we define and develop a metric structure on the ``space of networks". The crux of this construction is that many of the critical theoretical developments in \cite{dgh-sm, dghlp, sturm2012space} rely on measure-theoretic properties and not metric properties, hence they extend to the ambient space of networks. Certain interpretations and results cannot carry over: typically these are the statements involving volumes of open balls, which are hard to define in the asymmetric setting. The main algorithms of \cite{pcs16, s16} for computing local minima of the ERGW objective do carry over to the network setting. 

To complement these algorithms, we adapt ideas from \cite{dgh-sm, dghlp} to obtain network invariants/features that yield a hierarchy of lower bounds on the network GW distance. The lower bounds arise from satisfying a certain stability property, and are computed by solving (at most) a linear program. In experiments, we focus particularly on the (\ref{inv:TLB}) from \cite{dgh-sm}.

We strengthen some of the inequalities in the lower bound hierarchy to equalities (Theorem \ref{thm:stab-push}). As a consequence, we see that the (\ref{inv:TLB}), which involves solving an ensemble of OT problems over a Polish space $X\times Y$, can be computed by solving OT problems over $\R$ (\ref{inv:RTLB}). These can be directly computed via closed-form solutions.  

We also define a network Gromov-Prokhorov (GP) distance, propose a new class of invariants (the ``sublevel/superlevel size functions"), and use the GP distance to show that these new invariants satisfy a notion of \emph{interleaving} stability typically arising in the field of applied topology. We exhibit the theoretical utility of these invariants by using them to distinguish between spheres of different dimensions.

Finally, we illustrate our constructions on some highly asymmetric networks (both simulated and real). Our code and datasets are available on \url{https://github.com/samirchowdhury/GWnets}.

\subsection{Organization of the paper} 
In the following section, we define some notation and terms that will be used throughout the paper. \S\ref{sec:dist-net} contains details about couplings and the network Gromov-Wasserstein and Gromov-Prokhorov distances. In \S\ref{sec:lbounds} we present network invariants along with stability results. We conclude with experiments in \S\ref{sec:exp}. Appendix \ref{sec:computation} contains additional notes on computations.

\subsection{Notation and basic terminology}

We write $\R_+$ to denote the nonnegative reals.
The indicator function of a set $S$ is denoted $\mathbf{1}_S$.
Given a topological space $X$ (always a Polish space in this paper, and always equipped with the Borel $\s$-field $\borel(X)$), 
we will write $\prob(X)$ to denote the collection of Borel probability measures on $X$. 
The \emph{support} of $\mu_X \in \prob(X)$, denoted $\supp(\mu_X)$ (or $\supp(X)$ when the context is clear) is the set of $x\in X$ such that every open neighborhood of $x$ has positive measure. Unless specified otherwise, we will always deal with fully supported measures. 
The Lebesgue measure on the reals will be denoted by $\leb$.

The product $\s$-field on $X \times Y$, denoted $\borel(X\times Y)$, is defined as the $\s$-field generated by the measurable rectangles $A \times B$, where $A \in \borel(X)$ and $B \in \borel(Y)$. 
The product measure $\mu_X \otimes \mu_Y$ is defined on the measurable rectangles by writing 
\[\mu_X \otimes \mu_Y(A \times B) := \mu_X(A)\mu_X(B), \text{
for all $A \in \borel(X)$ and for all $B \in \borel(Y)$.}\]

Given a Borel space $(X,\mu_X)$ and a Borel measurable function $f:X \r \R$, we write $\norm{f}_p:= (\myint |f|^p \diff \mu_X)^{1/p}$ for $p \in [1,\infty)$, and $\norm{f}_\infty:=\inf\{M \in [0,\infty] : \mu_X(|f|> M) = 0\}$ for $p = \infty$. For each $p \in [1,\infty]$, $L^p = L^p(\mu_X)$ consists of the Borel measurable functions $f$ with $\lnorm f \rnorm_p < \infty$.

Given a measurable real-valued function $f: X \r \R$ and $t \in \R$, we will occasionally write $\{f \leq t \}$ to denote the set $\{ x \in X : f(x) \leq t \}$. 

Given $(X,\mu_X)$, $Y$, and a Borel-measurable map $f: X \r Y$, the pushforward of $\mu_X$ via $f$ is the measure defined by $f_*\mu(A):= \mu(f\inv[A])$ for any measurable subset of $Y$.

\section{The structure of measure networks} \label{sec:dist-net}

We will always assume that our measures are fully supported, unless explicitly said otherwise. 

\subsection{Networks and isomorphism}

\begin{definition}
A \emph{(measure) network} is a triple $(X,\w_X,\mu_X)$ where $X$ is Polish, $\mu_X$ is a fully supported Borel probability measure, and $\w_X$ is a bounded measurable function on $X^2$. The naming convention arises from the case when $X$ is finite; in such a case, we can view the pair $(X,\w_X)$ as a complete directed graph with asymmetric real-valued edge weights that is further equipped with node significance values given by $\mu_X$, cf. Figure \ref{fig:nets-header}. Accordingly, the points of $X$ are called \emph{nodes}, pairs of nodes are called \emph{edges}, and $\w_X$ is called the \emph{edge weight function} of $X$. The collection of all measure networks will be denoted $\Nm$.
\end{definition}

\begin{remark}[Network data] 
A large class of objects---including metric spaces, manifolds (Riemannian or Finslerian), and similarity/kernel matrices \cite{pcs16}---can be viewed as networks. Network datasets arising in the sciences typically satisfy the regularity assumptions needed to fit the preceding definition.

We point out one caveat: network datasets in the wild are often incomplete, i.e. $\w_X$ is not fully defined on $X\times X$. In such cases, one needs to preprocess the data (see e.g. \cite{kumar2016edge}) to make it fit within our framework. In many other use cases, however, network datasets are complete by construction. For example, in gene regulatory network inference \cite{sanguinetti2019gene}, the only data that can be measured are gene expression levels. In the corresponding network, the nodes are genes and the edge weights are gene dependencies that are \emph{estimated} from the expression levels. The resulting edge weight function is thus completely determined. 
\end{remark}

\begin{remark} Sturm has studied symmetric, $L^2$ versions of measure networks (called \emph{gauged measure spaces}) in \cite{sturm2012space}, and we point to his work as an excellent reference on the geometry of such spaces. Our motivation comes from studying network datasets, hence the difference in our naming conventions. 

\end{remark}

When defining any type of distance between networks, as we eventually will, it is necessary to first decide which networks should be viewed as being at 0-distance. We make these choices now.
The information contained in a network should be preserved under relabeling. Additionally, if a node is split into multiple nodes, all with the same incoming and outgoing edge weights, the information in the network remains unchanged. Conversely, if multiple nodes have the same incoming/outgoing edge weights, then they can be merged together without information loss. We formalize these ideas via the following notions of \emph{isomorphism}.

\begin{definition}[Strong isomorphism]
To say $(X,\w_X,\mu_X), (Y,\w_Y,\mu_Y) \in \Nm$ are \emph{strongly isomorphic} means that there exists a Borel measurable bijection $\ph: X \r Y$ (with Borel measurable inverse $\ph\inv$) such that 
\begin{itemize}
\item $\w_X(x,x') = \w_Y(\ph(x),\ph(x'))$ for all $x, x' \in X$, and
\item $\ph_*\mu_X = \mu_Y$.
\end{itemize}
We will denote a strong isomorphism between measure networks by $X \cong^s Y$.
\end{definition}
 
 The following definition is a relaxation of strong isomorphism.

\begin{definition}[Weak isomorphism]
\label{defn:weak-isom}
$(X,\w_X,\mu_X), (Y,\w_Y,\mu_Y) \in \Ngen$ are \emph{weakly isomorphic}, denoted $X\cong^w Y$, if there is a Borel probability space $(Z,\mu_Z)$ with measurable maps $f:Z \r X$ and $g:Z \r Y$ such that 
\begin{itemize}
\item $f_*\mu_Z = \mu_X$, $g_*\mu_Z = \mu_Y$, and 
\item $\norm{f^*\w_X - g^*\w_Y}_\infty = 0$.
\end{itemize}
Here $f^*\w_X:Z \times Z \r \R$ is the pullback weight function given by the map $(z,z') \mapsto \w_X(f(z),f(z'))$. The map $g^*\w_Y$ is defined analogously. Note that these pullbacks are measurable. Figure \ref{fig:wisom-nets} provides an illustration.

\end{definition}

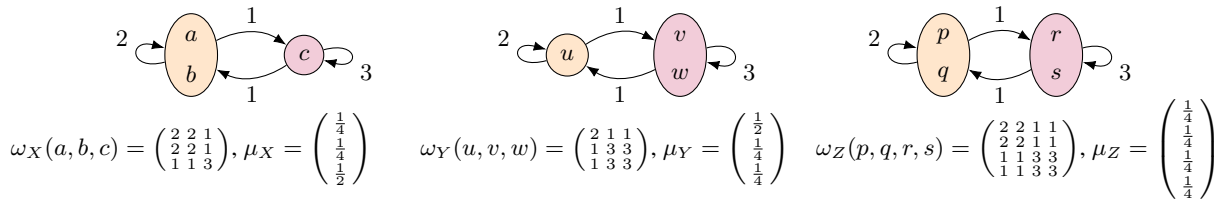
\begin{figure}[b]
\begin{center}
\begin{tikzpicture}[every node/.style={font=\footnotesize}]
\node[draw,ellipse, fill=orange!20, minimum width= 7mm,minimum height= 11 mm] (4) at (-1,-0.25) {};
\node[inner sep = 0pt, minimum size = 3mm] (1) at (-1,0){$a$};
\node[inner sep = 0pt, minimum size = 3mm] (2) at (-1,-0.5){$b$};
\node[draw,circle,fill=purple!20] (3) at (0.5,-0.25){$c$};


\node[draw,circle,fill=orange!20] (5) at (4,-0.25){$u$};
\node[draw,fill=purple!20, ellipse, minimum width= 7mm,minimum height= 11 mm] (8) at (5.5,-0.25) {};
\node[inner sep = 0pt, minimum size = 3mm,] (6) at (5.5,0){$v$};
\node[inner sep = 0pt, minimum size = 3mm] (7) at (5.5,-0.5){$w$};


\node[draw,fill=orange!20,ellipse, minimum width= 7mm,minimum height= 11 mm] (9) at (9,-0.25){};
\node[draw,fill=purple!20,ellipse, minimum width= 7mm,minimum height= 11 mm] (10) at (10.5,-0.25){};
\node[inner sep = 0pt, minimum size = 3mm] (11) at (9,0){$p$};
\node[inner sep = 0pt, minimum size = 3mm] (12) at (9,-0.5){$q$};
\node[inner sep = 0pt, minimum size = 3mm] (13) at (10.5,0){$r$};
\node[inner sep = 0pt, minimum size = 3mm] (14) at (10.5,-0.5){$s$};

\path[->] (4) edge [loop left, min distance = 5mm] node[above left]{$2$}(4);
\path[->] (3) edge [loop right, min distance = 5mm] node[below right]{$3$}(3);
\path[->] (4) edge [bend left] node[above]{$1$} (3);
\path[->] (3) edge [bend left] node[below]{$1$} (4);

\path[->] (5) edge [loop left, min distance = 5mm] node[above left]{$2$}(5);
\path[->] (8) edge [loop right, min distance = 5mm] node[below right]{$3$}(8);
\path[->] (5) edge [bend left] node[above]{$1$} (8);
\path[->] (8) edge [bend left] node[below]{$1$} (5);

\path[->] (9) edge [loop left, min distance = 5mm] node[above left]{$2$}(9);
\path[->] (10) edge [loop right, min distance = 5mm] node[below right]{$3$}(10);
\path[->] (9) edge [bend left] node[above]{$1$} (10);
\path[->] (10) edge [bend left] node[below]{$1$} (9);

\node at (-1,-1.5) {$\w_X(a,b,c) = \mattres{2}{2}{1}{2}{2}{1}{1}{1}{3}$, 
$\mu_X = \vecthree{\frac{1}{4}}{\frac{1}{4}}{\frac{1}{2}}$};
\node at (4.5,-1.5) {$\w_Y(u,v,w) = \mattres{2}{1}{1}{1}{3}{3}{1}{3}{3}$, 
$\mu_Y = \vecthree{\frac{1}{2}}{\frac{1}{4}}{\frac{1}{4}}$};

\node at (10,-1.5) {$\w_Z(p,q,r,s) = 
\left(
\begin{smallmatrix}
2 & 2 & 1 & 1\\
2 & 2 & 1 & 1\\
1 & 1 & 3 & 3\\
1 & 1 & 3 & 3\\
\end{smallmatrix}
\right)
$, 
$\mu_Z = \vecfour{\frac{1}{4}}{\frac{1}{4}}{\frac{1}{4}}{\frac{1}{4}}$};

\end{tikzpicture}
\caption{Weakly isomorphic networks $X,Y,Z$. Note that $Z$ maps surjectively onto $X$ and $Y$, and this surjection induces compatible pushforward measures and pullback weight functions, as required by Definition \ref{defn:weak-isom}.}
\label{fig:wisom-nets}
\end{center}
\vspace{-.2in}
\end{figure}

\begin{remark}[Interpretation for real data] 
According to the notion of weak isomorphism, two nodes $x,x'$ are informally the same if they have the same ``internal perception", i.e. $\w_X(x,x) = \w_X(x,x') = \w_X(x',x) = \w_X(x',x')$, and the same external perception, i.e. all the incoming and outgoing edge weights are the same. A relaxation would be to say that $x,x'$ are $\e$-similar if, for $\e > 0$,  
\[\max(\lnorm \w_X(x,\cdot) - \w_X(x',\cdot) \rnorm_\infty, \lnorm \w_X(\cdot,x) - \w_X(\cdot,x') \rnorm_\infty) < \e.\] 
The network stochastic block model in \S\ref{sec:sbm-net} describes networks that admit partitions into $\e$-similar blocks. 
\end{remark}

\begin{example}\label{ex:simple-networks} Networks with one or two nodes will be very instructive in providing examples and counterexamples, so we introduce them now with some special terminology.
\begin{itemize}
\item By $N_1(a)$ we will refer to the network with one node $X = \set{p}$, a weight $\w_X(p,p) = a$, and the Dirac measure $\mu_X = \d_p$.

\item By $N_2(\mattwo{a}{b}{c}{d}, \a, \b)$ we will mean a two-node network with node set $X = \set{p,q}$, and weights and measures given as follows:
\begin{align*}
\w_X(p,p) = a & \qquad \mu_X(\set{p}) = \a\\
\w_X(p,q) = b & \qquad \mu_X(\set{q}) = \b\\
\w_X(q,p) = c \\
\w_X(q,q) = d 
\end{align*}

\item Given a $k$-by-$k$ matrix $\Sigma\in \R^{k\times k}$ and a $k\times 1$ vector $v\in\R^{k}_+$ with sum $1$, we automatically obtain a network on $k$ nodes that we denote as $N_k(\Sigma,v)$. Notice that $N_k(\Sigma,v)\cong^s N_\ell(\Sigma',v')$ if and only if $k=\ell$ and there exists a permutation matrix $P$ of size $k$ such that $\Sigma'=P\,\Sigma\,P^T$ and $Pv = v'$.
\end{itemize}
\end{example}

\noindent
\textbf{Notation.} Even though $\mu_X$ takes sets as its argument, we will often omit the curly braces and use $\mu_X(p,q,r)$ to mean $\mu_X(\set{p,q,r})$. 

We wish to define a notion of distance on $\Nm$ that is compatible with isomorphism. A natural analog is the Gromov-Wasserstein distance defined between metric measure spaces \cite{dgh-sm}. To adapt that definition for our needs, we first recall the definition of a measure coupling. 

\subsection{Couplings and the distortion functional}

Let $(X,\w_X,\mu_X), (Y,\w_Y,\mu_Y)$ be two measure networks. A \emph{coupling} between these two networks is a probability  measure $\mu$ on $X\times Y$ with marginals $\mu_X$ and $\mu_Y$, respectively. Stated differently, couplings satisfy the following property:
\[\mu(A \times Y) = \mu_X(A) \; \text{and} \;
\mu(X \times B) = \mu_Y(B), \text{
for all $A \in \borel(X)$ and for all $B \in \borel(Y)$.}\]
The collection of all couplings between $(X,\w_X,\mu_X)$ and $(Y,\w_Y,\mu_Y)$ will be denoted $\coup(\mu_X,\mu_Y)$, abbreviated to $\coup$ when the context is clear. Couplings are also referred to as \emph{transport plans}.

\begin{example}[Product coupling] 
\label{ex:product-coupling}
Let $(X,\w_X,\mu_X), \ (Y,\w_Y,\mu_Y) \in \Nm$. The set $\coup(\mu_X,\mu_Y)$ is always nonempty, because the \emph{product measure} $\mu:= \mu_X \otimes \mu_Y$ is always a coupling between $\mu_X$ and $\mu_Y$. 
\end{example}

\begin{example}[1-point coupling]\label{ex:1-pt-coupling} Let $(X,\w_X,\mu_X) \in \Nm$, and let $Y = N_1(a)$ be a network on a single point $\{p\}$. Then there exists a unique coupling $\mu = \mu_X \otimes \d_p$ between $\mu_X$ and $\delta_p$.

\end{example} 

\begin{example}[Diagonal coupling]
\label{ex:diag-coupling}
Let $(X,\w_X,\mu_X) \in \Nm$. The \emph{diagonal coupling $\Delta$} between $\mu_X$ and itself is the transport plan that sends each point to itself, and is defined by writing 
\[\Delta(A\times B) := \int_{X} \mathbf{1}_{A\times B}(x,x) \, d\mu_X(x) \qquad
\text{ for all } A, B\in \borel(X).\]
To see that this is a coupling, let $A\in \borel(X)$. Then,
\[\Delta(A\times X) = \int_{X} \mathbf{1}_{A\times X}(x,x) \, d\mu_X(x) = \int_{X}\1_{A}(x) \, d\mu_X(x) = \mu_X(A),\]
and similarly $\Delta(X\times A) = \mu_X(A)$. Thus $\Delta \in \coup(\mu_X,\mu_X)$. 
\end{example}

Now we turn to the notion of the \emph{distortion} of a coupling.
Let $(X,\w_X,\mu_X),(Y,\w_Y,\mu_Y)$ be two measure networks. 
Next let $\mu\in\coup(\mu_X,\mu_Y)$, and consider the probability space $(X\times Y)^2$ equipped with the product measure $\mu\otimes \mu$. For each $p \in [1,\infty]$ the $p$-distortion of $\mu$ is defined as $\dis_p(\mu) := \lnorm \w_X - \w_Y \rnorm_p$. For $p \in [1,\infty)$, this is written as:
\begin{align*}
\dis_p(\mu) 
& = \left( \int_{X\times Y}\int_{X\times Y}|\w_X(x,x') - \w_Y(y,y')|^p \diff \mu(x,y) \diff\mu(x',y') \right)^{1/p}.
\end{align*}
For $p = \infty$, this becomes:
\[\dis_p(\mu) := \esssup  |\w_X - \w_Y|.\]

We end by introducing the Wasserstein distance \cite[\S7]{ambrosio2008gradient}, which metrizes the topology of narrow convergence that we introduce below. Let $(X,d_X)$ be a Polish space, let $p \in [1,\infty]$, and let $\mu,\nu \in \prob(X)$ be such that $\lnorm d_X(x_0,\cdot) \rnorm_{L^p(\t)} < \infty$ for $\t \in \{ \mu,\nu \}$ and some $x_0 \in X$. The \emph{pth Wasserstein distance} between $\mu,\nu$ is defined to be: 
\[W_p(\mu,\nu) := \inf_{\t \in \coup(\mu,\nu)} \lnorm d_X \rnorm_{L^p(\t)}.\]

\subsection{Optimality of couplings in the network setting}

We now collect some results about probability spaces.
Let $X$ be a Polish space. A subset $P \subseteq \prob(X)$ is said to be \emph{tight} if for all $\e > 0$, there is a compact subset $K_\e \subseteq X$ such that $\mu_X(X\setminus K_\e) \leq \e$ for all $\mu_X \in P$. 

A sequence $(\mu_n)_{n\in \N} \in \prob(X)^\N$ is said to \emph{converge narrowly} to $\mu_X \in \prob(X)$ if 
\[\lim_{n \r \infty} \int_X f \, d\mu_n = \int_X f \,d\mu_X
\qquad \text{ for all } f \in C_b(X),\]
the space of continuous, bounded, real-valued functions on $X$. Narrow convergence is induced by a distance \cite[Remark 5.1.1]{ambrosio2008gradient}, in particular by $W_p$ when $X$ is bounded, hence the convergent sequences in $\prob(X)$ completely determine a topology on $\prob(X)$. This topology on $\prob(X)$ is called the \emph{narrow topology}. In some references \cite{sturm2012space}, narrow convergence (resp. narrow topology) is called \emph{weak convergence} (resp. \emph{weak topology}).

A further consequence of having a metric on $\prob(X)$ \cite[Remark 5.1.1]{ambrosio2008gradient} is that singletons are closed. This simple fact will be used below.

\begin{theorem}[Prokhorov, \cite{ambrosio2008gradient} Theorem 5.1.3] 
Let $X$ be a Polish space. Then $P\subseteq \prob(X)$ is tight if and only if it is relatively compact, i.e. its closure is compact in $\prob(X)$.
\end{theorem}

\begin{lemma}[Lemma 4.4, \cite{villani2008optimal}]
\label{lem:tightness-couplings}
Let $X,Y$ be two Polish spaces, and let $P_X \subseteq \prob(X)$, $P_Y \subseteq \prob(Y)$ be tight in their respective spaces. Then the set $\coup(P_X,P_Y) \subseteq \prob(X\times Y)$ of couplings with marginals in $P_X$ and $P_Y$ is tight in $\prob(X\times Y)$.
\end{lemma}

\begin{lemma}[Compactness of couplings; Lemma 1.2, \cite{sturm2012space}]
\label{lem:compactness-couplings}
Let $X,Y$ be two Polish spaces. Let $\mu_X \in \prob(X)$, $\mu_Y \in \prob(Y)$. Then $\coup(\mu_X,\mu_Y)$ is compact in $\prob(X\times Y)$.
\end{lemma}

\begin{proof}
The singletons $\{\mu_X\}$, $\{ \mu_Y\}$ are closed and of course compact in $\prob(X)$, $\prob(Y)$. Hence by Prokhorov's theorem, they are tight. Now consider $\coup(\mu_X,\mu_Y) \subseteq \prob(X\times Y)$. Since this is obtained by intersecting the preimages of the continuous projections onto the marginals $\mu_X$ and $\mu_Y$, we know that it is closed. Furthermore, $\coup(\mu_X,\mu_Y)$ is tight by Lemma \ref{lem:tightness-couplings}. Then by another application of Prokhorov's theorem, it is compact.  \qedhere
\end{proof}

The following lemma appeared for symmetric weight functions in the $L^2$ case in \cite{sturm2012space}, along with a slightly different proof using parametrizations by the unit interval. The proof is actually simpler in the network setting because we do not need to enforce symmetry of the approximating functions.

\begin{lemma}[Continuity of the distortion functional]
\label{lem:cont-distortion-intervals}
Let $1\leq p < \infty$, and let $(X,\w_X,\mu_X), (Y,\w_Y,\mu_Y) \in \Nm$. The distortion functional $\dis_p$ is continuous on $\coup(\mu_X,\mu_Y)$. For $p = \infty$, $\dis_\infty$ is lower semicontinuous. 
\end{lemma}

\begin{proof}

First suppose $p \in [1,\infty)$.
We will construct a sequence of continuous functionals that converges uniformly to $\dis_p$. Since the uniform limit of continuous functions is continuous, this will show that $\dis_p$ is continuous.

Bounded continuous functions are dense in $L^p$ (in our setting of Polish spaces with finite measures, see e.g. \cite[\S7.2]{folland1999real}), so for each $n \in \N$, we pick continuous, bounded functions $\w_X^n \in L^p(\mu_X^{\otimes 2})$ and $\w_Y^n \in L^p(\mu_Y^{\otimes 2})$ such that 
\[\lnorm \w_X - \w_X^n \rnorm_{L^p(\mu_X \otimes \mu_X)} \leq 1/n, \qquad
\lnorm \w_Y - \w_Y^n \rnorm_{L^p(\mu_Y \otimes \mu_Y)} \leq 1/n.\]

For each $n \in \N$, define the functional $\dis_p^n:\coup(\mu_X,\mu_Y) \r \R_+$ by $\dis_p^n(\nu):= \lnorm \w_X^n - \w_Y^n \rnorm_{L^p(\nu \otimes \nu)}$.
Note that $\lbar \w_X^n - \w_Y^n \rbar^p \in C_b((X\times Y)^2)$. 

We claim that $\dis_p^n$ is continuous. Since the narrow topology on $\prob(X \times Y)$ is induced by a distance \cite[Remark 5.1.1]{ambrosio2008gradient}, it suffices to show sequential continuity. Let $\nu \in \coup(\mu_X,\mu_Y)$, and let $(\nu_m)_{m\in \N}$ be a sequence in $\coup(\mu_X,\mu_Y)$ converging narrowly to $\nu$. Then in fact $\nu_m \otimes \nu_m$ converges narrowly to $\nu\otimes \nu$ \cite[Theorem 2.8]{billingsley1999convergence}. Thus we have 
\begin{align*}
\lim_{m\r \infty} \dis_p^n(\nu_m) 
& = \lim_{m\r \infty} \lp \int_{ (X\times Y)^2 } \lbar \w^n_X- \w^n_Y\rbar^p d\nu_m \otimes d\nu_m\rp^{1/p}\\
&= \lp \int_{ (X\times Y)^2 } \lbar \w^n_X- \w^n_Y\rbar^p d\nu \otimes d\nu\rp^{1/p}
 = \dis_p^n(\nu).
\end{align*}
Here the second equality follows from the definition of convergence in the narrow topology and the fact that the integrand is bounded and continuous. This shows sequential continuity (hence continuity) of $\dis_p^n$. 

Finally, we show that $(\dis_p^n)_{n \in \N}$ converges to $\dis_p$ uniformly. Let $\mu \in \coup(\mu_X,\mu_Y)$. Then,

\begin{align*}
\lbar \dis_p(\mu) - \dis_p^n(\mu) \rbar &= \lbar \lnorm \w_X - \w_Y \rnorm_{L^p(\nu \otimes \nu)} - \lnorm \w_X^n - \w_Y^n \rnorm_{L^p(\nu \otimes \nu)} \rbar \\
& \leq \lnorm \w_X - \w_X^n \rnorm_{L^p(\mu_X \otimes \mu_X)} + \lnorm \w_Y - \w_Y^n \rnorm_{L^p(\mu_Y \otimes \mu_Y)} 
 \leq 2/n.
\end{align*}
But $\mu \in \coup(\mu_X,\mu_Y)$ was arbitrary. This shows that $\dis_p$ is the uniform limit of continuous functions, hence is continuous. Here the first inequality followed from Minkowski's inequality.  

Now suppose $p = \infty$. Let $\mu \in \coup(\mu_X,\mu_Y)$ be arbitrary. Recall that because we are working over probability spaces, Jensen's inequality can be used to show that for any $1 \leq q \leq r < \infty$, we have $\dis_q(\mu) \leq \dis_r(\mu)$. Moreover, we have $\lim_{q \r \infty}\dis_q(\mu) = \dis_\infty(\mu)$. The supremum of a family of continuous functions is lower semicontinuous. In our case, $\dis_\infty = \sup \{\dis_q : q \in [1,\infty)\}$, and we have shown above that all the functions in this family are continuous. Hence $\dis_\infty$ is lower semicontinuous. \end{proof}

\begin{definition}[Optimal couplings] Let $(X,\w_X,\mu_X),(Y,\w_Y,\mu_Y) \in \Nm$, and let $p \in [1,\infty]$. A coupling $\mu \in \coup(\mu_X,\mu_Y)$ is \emph{optimal} if $\dis_p(\mu) = \inf_{\nu \in \coup(\mu_X,\mu_Y)}\dis_p(\nu)$.
\end{definition}

\begin{theorem} 
\label{thm:opt-coup}
Let $(X,\w_X,\mu_X)$ and $(Y,\w_Y,\mu_Y)$ be two measure networks, and let $p \in [1,\infty]$. Then there exists an optimal coupling, i.e. a minimizer for $\dis_p(\cdot)$ in $\coup(\mu_X,\mu_Y)$.
\end{theorem}

\begin{proof}
The result follows from Lemmas \ref{lem:compactness-couplings} and \ref{lem:cont-distortion-intervals}, because lower semicontinuity and compactness are sufficient to guarantee that $\dis_p$ achieves its infimum on $\coup(\mu_X,\mu_Y)$, for any $p \in [1,\infty]$. \qedhere

\end{proof}

\subsection{The network Gromov-Wasserstein distance}
For each $p \in [1,\infty]$, we define:
\[\dnp(X,Y):=\frac{1}{2}\min_{\mu\in\coup(\mu_X,\mu_Y)}\dis_p(\mu) 
\qquad \text{for each } (X,\w_X,\mu_X),(Y,\w_Y,\mu_Y) \in \Nm.
\]
Here we implicitly use Theorem \ref{thm:opt-coup} to write $\min$ instead of $\inf$. As we will see below, $\dnp$ is a legitimate pseudometric on $\Nm$. The structure of $\dnp$ is analogous to a formulation of the \emph{Gromov-Wasserstein distance} between metric measure spaces \cite{dghlp, sturm2012space}. 

\begin{remark}[Boundedness of $\dnp$]
Recall from Example \ref{ex:product-coupling} that for any $X,Y\in \Nm$, $\coup(\mu_X,\mu_Y)$ always contains the product coupling, and is thus nonempty. A consequence is that $\dnp(X,Y)$ is bounded for any $p \in [1,\infty]$. Indeed, by taking the product coupling $\mu:= \mu_X \otimes \mu_Y$ we have $\dnp(X,Y)\leq \frac{1}{2}\dis_p(\mu) < \infty$.

\end{remark}

In some simple cases, we obtain explicit formulas for computing $\dnp$.

\begin{example}[Easy examples of $\dnp$]
\label{ex:1-pt-dnp}
Let $a, b \in \R$ and consider the networks $N_1(a)$ and $N_1(b)$. The unique coupling between the two networks is the product measure $\mu = \d_x \otimes \d_y$, where we understand $x,y$ to be the nodes of the two networks. Then for any $p \in [1, \infty]$, we obtain:
\[ \dnp(N_1(a),N_1(b)) = \frac{1}{2}\dis_p(\mu) = |\w_{N_1(a)}(x,x) - \w_{N_1(b)}(y,y)| = |a-b|.\] 

Let $(X,\omega_X,\mu_X)\in \Nm$ be any network and let $N_1(a)=(\{y\},a)$ be a network with one node. Once again, there is a unique coupling $\mu = \mu_X\otimes \d_y$ between the two networks. For any $p \in [1,\infty)$, we obtain:
\[\dnp(X,N_1(a))=\frac{1}{2}\dis_p(\mu) = \frac{1}{2} \lp\int_{X}\int_X \lbar \w_X(x,x') - a \rbar^p d\mu_X(x)d\mu_X(x') \rp^{1/p}.\]
For $p = \infty$, we have $\dnp(X,N_1(a)) = \esssup \left( \frac{1}{2}|\w_X - a| \right).$
\end{example}

\begin{remark}\label{rem:weak-isom} $\dnp$ is not necessarily a metric modulo strong isomorphism. This can be seen from Figure \ref{fig:wisom-nets}.
\end{remark}

The definition of $\dnp$ is sensible in the sense that it captures the notion of a distance:

\begin{theorem}\label{thm:dnp}
For each $p \in [1,\infty]$, $\dnp$ is a pseudometric on $\Nm$.
\end{theorem}

\begin{proof}[Proof of Theorem \ref{thm:dnp}] 

Let $(X,\w_X,\mu_X), (Y,\w_Y,\mu_Y), (Z,\w_Y,\mu_Y) \in \Nm$.
It is clear that $\dnp(X,Y) \geq 0$. 
Taking the diagonal coupling (see Example \ref{ex:diag-coupling}) shows $\dnp(X,X) = 0$. 
For symmetry, notice that for any $\mu \in \coup(\mu_X,\mu_Y)$, 
we can define $\widetilde{\mu} := f_*\mu$, where $f:X\times Y \r Y\times X$ is the map $(x,y) \mapsto (y,x)$.
Then $\dis_p(\mu)=\dis_p(\widetilde{\mu})$, and this will show $\dnp(X,Y)=\dnp(Y,X)$. Note that we are overloading notation here: there are implicitly two $\dis_p$ functions, with different domains, for $\mu$ and $\widetilde{\mu}$, respectively. 

Finally, we need to check the triangle inequality. Let $\mu_{12} \in \coup(\mu_X,\mu_Y)$ and $\mu_{23} \in \coup(\mu_Y,\mu_Z)$ be couplings such that $2\dnp(X,Y) = \dis_p(\mu_{12}) $ and 
$2\dnp(Y,Z) = \dis_p(\mu_{23}) $ (using Theorem \ref{thm:opt-coup}). 
By the standard gluing lemma (Lemma 1.4 in \cite{sturm2012space}, also Lemma 7.6 in \cite{villani2003topics}), we obtain a probability measure $\mu \in \prob(X\times Y\times Z)$ with marginals $\mu_{12}, \mu_{23}$, and a marginal $\mu_{13}$ that is a coupling between $\mu_X$ and $\mu_Z$. This coupling is not necessarily optimal. Then we have:
\begin{align*}
2\dnp(X,Z) &\leq \dis_p(\mu_{13})\\
&=\lnorm \w_X - \w_Y + \w_Y - \w_Z\rnorm_{L^p(\mu \otimes \mu)}\\
&\leq \lnorm \w_X - \w_Y \rnorm_{L^p(\mu \otimes \mu)} + \lnorm \w_Y - \w_Z \rnorm_{L^p(\mu \otimes \mu)}\\
&= \lnorm \w_X - \w_Y \rnorm_{L^p(\mu_{12} \otimes \mu_{12})} + \lnorm \w_Y - \w_Z \rnorm_{L^p(\mu_{23} \otimes \mu_{23})}
 = 2 \dnp(X,Y)+2\dnp(Y,Z).
\end{align*}
The second inequality above follows from Minkowski's inequality. 
This proves the triangle inequality. \end{proof}

\begin{remark}
This result and its proof are analogous to the related results for gauged and metric measure spaces \cite{dghlp, sturm2012space}. The observation here is that the metric structure on $\Ngen$ is not inherited from its elements, but is rather enforced by the structure of $\dnp$. This is in contrast, for example, to the Wasserstein distance $W_p$, which inherits its metric structure from an underlying metric space.
\end{remark}

It remains to discuss the precise pseudometric structure of $\dnp$. The following result is analogous to a statement about \emph{homomorphisms} in \cite{sturm2012space}; again, the proof is purely measure-theoretic and hence applies to the asymmetric setting.

\begin{theorem}[Pseudometric structure of $\dnp$]
\label{thm:dnp-pseudo}
Let $(X,\w_X,\mu_X), (Y,\w_Y,\mu_Y) \in \Ngen$, and let $p \in [1,\infty]$. Then $\dnp(X,Y) = 0$ if and only if $X \cong^w Y$. 
\end{theorem}

\begin{proof}[Proof of Theorem \ref{thm:dnp-pseudo}]
Fix $p \in [1,\infty)$. For the backward direction, suppose there exist $(Z,\mu_Z)$ and measurable maps $f:Z \r X$ and $g:Z \r Y$ satisfying the conditions of Definition \ref{defn:weak-isom}. Let $\mu:= (f,g)_*\mu_Z$. Then $\mu \in \coup(\mu_X,\mu_Y)$, and we have:
\begin{align*}
2\dnp(X,Y) \leq \lp \int_{(X\times Y)^2}  \lbar \w_X - \w_Y \rbar^p \diff\mu\diff\mu \rp^{1/p}
= \lp \int_{Z^2} \lbar f^*\w_X - g^*\w_Y\rbar^p \diff \mu_Z\diff\mu_Z\rp^{1/p} = 0.
\end{align*}
Here the first equality is by the change of variables formula. The case $p = \infty$ is similar.

For the forward direction, let $\mu \in \coup(\mu_X,\mu_Y)$ be an optimal coupling with $\dis_p(\mu) = 0$ (Theorem \ref{thm:opt-coup}). Define $Z:= X\times Y$, $\mu_Z:= \mu$. Then the projection maps $\pi_X:Z \r X$ and $\pi_Y:Z \r Y$ are measurable. We also have $(\pi_X)_*\mu = \mu_X$ and $(\pi_Y)_*\mu =\mu_Y$. Since $\dis_p(\mu) = 0$, we also have $\norm{(\pi_X)^*\w_X - (\pi_Y)^*\w_Y}_\infty = \norm{\w_X - \w_Y}_\infty = 0$. 

The $p=\infty$ case is proved analogously. This concludes the proof. \qedhere
\end{proof}

\begin{remark} 
A result analogous to Theorem \ref{thm:dnp-pseudo} holds for networks without measure equipped with a Gromov-Hausdorff-type network distance \cite{dn-part1}. The ``tripod structure" $X \leftarrow Z \r Y$ described above is much more difficult to obtain in the setting of \cite{dn-part1}. This highlights an advantage of the measure-theoretic setting of the current paper. 
\end{remark}

\subsection{Additional constructions}

We briefly digress to discuss some theoretical connections to the framework presented above. The first of these is the notion of parametrization, which is used in the setting of mm-spaces to define Gromov's box distance \cite{gromov-book}. The second is an explicit development of an alternative distance between networks based on the Gromov-Prokhorov distance between mm-spaces \cite{greven2009convergence}. This in turn leads to interesting and novel lower bounds on the $\dninf$-distance between spheres (see \S\ref{sec:sphere-lbounds}).

\subsubsection{Interval representation}
\label{sec:interval-rep}

We now record a standard result about mm-spaces that remains valid in the network setting. Let $(X,\w_X,\mu_X) \in \Nm$. Because $X$ is Polish and $\mu_X(X) = 1$, the pair $(X,\mu_X)$ admits a \emph{parameter}, i.e. a (not necessarily unique) surjective Borel-measurable map $\rho:I=[0,1] \r X$ such that $\rho_*\leb = \mu_X$ \cite[Lemma 4.2]{shioya2016metric}. Here $\leb$ denotes Lebesgue measure. By pulling back $\w_X$, we get a triple $(I, \rho^*\w_X, \leb) \in \Nm$. Note that by construction, $(X,\w_X,\mu_X)$ is weakly isomorphic to its parameter. 

Parametrizations allow one to define a version of Gromov's box distance \cite{gromov-book} for networks. Computing the box distance leads to difficult combinatorial problems and is not the focus of this paper, but we point to it as a source of interesting theoretical problems.

In parametrized form, a network is a measurable, integrable function on the unit square. If the edge weight function is normalized and centered to be in $[0,1]$, then a network corresponds to a \emph{graphon} \cite{lovasz2012large}.

\subsubsection{The network Gromov-Prokhorov distance}

We now formulate a network distance analogous to the \emph{Gromov-Prokhorov} distance between mm-spaces \cite{greven2009convergence}. This will be used to prove subsequent results. 

Let $\a \in [0,\infty)$. For any $(X,\w_X,\mu_X)$, $(Y,\w_Y,\mu_Y) \in \Nm$, we write $\coup := \coup(\mu_X,\mu_Y)$ and define:
\[ \dgp(X,Y) :=\frac{1}{2} \inf_{\mu \in \coup} \inf \{ \e > 0 : \mu\otimes\mu \lp \{x,y,x',y' \in (X\times Y)^2 : \lbar \w_X(x,x') - \w_Y(y,y') \rbar \geq \e \} \rp \leq \a\e\}.  \]

\begin{theorem}
\label{thm:gp-metric}
For each $\a \in [0,\infty)$, $\dgp$ is a pseudometric on $\Nm$.
\end{theorem}
\begin{proof}
Let $(X,\w_X,\mu_X),(Y,\w_Y,\mu_Y), (Z,\w_Z,\mu_Z) \in \Nm$. The proofs that $\dgp(X,Y) \geq 0$, $\dgp(X,X) = 0$, and that $\dgp(X,Y) = \dgp(Y,X)$ are analogous to those used in Theorem \ref{thm:dnp}. Hence we only check the triangle inequality. Let $\e_{XY} > 2\dgp(X,Y)$, $\e_{YZ} > 2\dgp(Y,Z)$, and let $\mu_{XY}, \mu_{YZ}$ be couplings such that 
\begin{align*}
\mu_{XY}^{\otimes 2} \lp \{(x,y,x',y') : \lbar \w_X(x,x') - \w_Y(y,y') \rbar \geq \e_{XY} \}\rp &\leq \a\e_{XY},\\
\mu_{YZ}^{\otimes 2} \lp \{(y,z,y',z') : \lbar \w_Y(y,y') - \w_Z(z,z') \rbar \geq \e_{YZ} \}\rp &\leq \a\e_{YZ}.
\end{align*}
For convenience, define:
\begin{align*}
A&:= \{((x,y,z),(x',y',z')) \in (X\times Y \times Z)^2 : \lbar \w_X(x,x') - \w_Y(y,y') \rbar \geq \e_{XY} \}\\
B&:= \{((x,y,z),(x',y',z')) \in (X\times Y \times Z)^2 : \lbar \w_Y(y,y') - \w_Z(z,z') \rbar \geq \e_{YZ} \}\\
C&:=\{((x,y,z),(x',y',z')) \in (X\times Y \times Z)^2 : \lbar \w_X(x,x') - \w_Z(z,z') \rbar \geq \e_{XY} + \e_{YZ}\}.
\end{align*}

Next let $\mu$ denote the probability measure with marginals $\mu_{XY}, \mu_{YZ}$, and a marginal $\mu_{XZ}\in \coup(\mu_X,\mu_Z)$ obtained from gluing $\mu_{XY}$ and $\mu_{YZ}$ (cf. Lemma 7.6 in \cite{villani2003topics}). We need to show: 
\[\mu_{XZ}^{\otimes 2}\lp (\pi_X,\pi_Z)(C) \rp \leq \a(\e_{XY} + \e_{YZ}).\]
To show this, it suffices to show $C \subseteq A \cup B$, because then we have $\mu^{\otimes 2}(C) \leq \mu^{\otimes 2}(A) + \mu^{\otimes 2}(B)$ and consequently 
\begin{align*}
\mu_{XZ}^{\otimes 2}\lp (\pi_X,\pi_Z)(C) \rp = \mu^{\otimes 2}(C) \leq \mu^{\otimes 2}(A) + \mu^{\otimes 2}(B) &= 
\mu_{XY}^{\otimes 2}\lp (\pi_X,\pi_Y)(A) \rp  + 
\mu_{YZ}^{\otimes 2}\lp (\pi_Y,\pi_Z)(B) \rp\\ 
&\leq \a(\e_{XY} + \e_{YZ}).
\end{align*}

Let $((x,y,z),(x',y',z')) \in (X\times Y\times Z)^2 \setminus (A \cup B)$. Then we have 
\[\lbar \w_X(x,x') - \w_Y(y,y') \rbar < \e_{XY} \text{ and }
\lbar \w_Y(y,y') - \w_Z(z,z') \rbar < \e_{YZ}.\]
By the triangle inequality, we then have:
\[ \lbar \w_X(x,x') - \w_Z(z,z') \rbar \leq 
\lbar \w_X(x,x') - \w_Y(y,y') \rbar + 
\lbar \w_Y(y,y') - \w_Z(z,z') \rbar < \e_{XY} + \e_{YZ}.\]
Thus $((x,y,z),(x',y',z')) \in (X\times Y\times Z)^2 \setminus C$. This shows $C \subseteq A \cup B$.

The preceding work shows that $2\dgp(X,Z) \leq \e_{XY} + \e_{YZ}$. Since $\e_{XY} > 2\dgp(X,Y)$ and $\e_{YZ} > 2\dgp(Y,Z)$ were arbitrary, it follows that $\dgp(X,Z) \leq \dgp(X,Y) + \dgp(Y,Z)$. \end{proof}

The next lemma follows by unpacking the definitions of the GP and GW distances.
\begin{lemma}[Relation between GP and GW]
\label{lem:gp-gw}
Let $(X,\w_X,\mu_X), (Y,\w_Y,\mu_Y) \in \Nm$. We always have:
\[\dgpzero(X,Y) = \dninf(X,Y).\]
\end{lemma}

\section{Invariants and lower bounds}
\label{sec:lbounds}

As already remarked, finite networks can be regarded as square matrices equipped with a probability measure on the columns (equivalently, the rows). This was the setting of \cite{pcs16}. 
Theorem \ref{thm:dnp} completes the theoretical justification behind using the GW distance to compare matrices, as carried out (at least for symmetric matrices) in \cite{pcs16}. 

We now study a variety of network invariants, which can also be thought of as network features. 
Informally, a network invariant is a compressed representation of the network satisfying the following compatibility property: if two networks are the same in the sense of $\dnp$, then their invariants should also be ``the same". 
The invariants we consider are functions $\iota: (\Ngen,\dnp) \r (\mc{I},d_{\mc{I}})$, where $(\mc{I},d_{\mc{I}})$ is some pseudometric space.  
Such invariants translate the original problem of computing GW over the ``space of networks" to computing $d_{\mc{I}}$ over spaces $\mc{I}$ with more regular geometry, e.g. the real line.
Translating the problem to a simpler space is done in a controlled manner. One such form of control is \emph{Lipschitz stability}: an invariant $\iota$ is \emph{Lipschitz-stable} if there exists a Lipschitz constant $L_{\iota}$ such that 
\[ d_{\mc{I}}(\iota(X),\iota(Y)) \leq L_{\iota} \cdot \dnp(X,Y), \tag*{for all $X,Y \in \Ngen.$}\]
In Section \ref{sec:lip-stable}, we present Lipschitz-stable invariants. In Section \ref{sec:int-stable}, we present a different notion of control that we refer to as \emph{interleaving stability} as well as associated invariants.

\begin{remark}
Asymmetry arises in a significant way in this section: for most of our network invariants, we obtain ``outgoing" and ``incoming" versions, based on our choice of functions $\w_X(x,\cdot)$ or $\w_X(\cdot,x)$. The network interpretation can be framed in terms of hubs (nodes with high outgoing edge weights) and authorities (nodes with high incoming edge weights) \cite{kleinberg1999authoritative}.
\end{remark}

\subsection{A hierarchy of lower bounds for $\dnp$}
\label{sec:lip-stable}

Following \cite{dgh-sm,dghlp}, we now produce a hierarchy of lower bounds for $\dnp$, namely the First Lower Bound (FLB), Second Lower Bound (SLB), and Third Lower Bound (TLB). These are obtained by linearizing the GW objective and/or pushing forward the problem into the real line. Each of these bounds itself has an associated pushforward into the real line, which we denote by adding a prefix $\R$-. 
Because we are in the asymmetric setting, the FLB, TLB, and their $\R$-versions decouple into ``incoming" and ``outgoing" versions. 
In Remark \ref{rem:conn-lb-invar}, we will show (using Theorem \ref{thm:stab-push}) that these lower bounds are in fact obtained as Lipschitz stability conditions on certain network invariants.

The hierarchy is illustrated in the following diagram, where the arrows indicate (possibly non-strictly) decreasing complexity. 

\begin{center}
\begin{tikzpicture}
\node (size) at (-3,0){$\size$};
\node (FLB) at (0,0){FLB};
\node (TLB) at (3,0){TLB};
\node (GW) at (6,0){GW};
\node (SLB) at (9,0){SLB};
\node (rFLB) at (0,-1){$\R$-FLB};
\node (rTLB) at (3,-1){$\R$-TLB};
\node (rSLB) at (9,-1){$\R$-SLB};

\draw (GW) edge[->] (SLB);
\draw (GW) edge[->] (TLB);
\draw (TLB) edge[->] (rTLB);
\draw (FLB) edge[->] (rFLB);
\draw (SLB) edge[->] (rSLB);
\draw (TLB) edge[->] (FLB);
\draw (FLB) edge[->] (size);
\end{tikzpicture}
\end{center}

We now derive these relationships. Let $p\in [1,\infty]$. Start by fixing one pair of coordinates in the $\dnp$ integrand. Then we obtain the \emph{(outgoing) joint eccentricity function} $\eccout_{p,X,Y}:X\times Y \r \R_+$ defined by
\begin{align}
\eccout_{p,X,Y}(s,t):= \inf_{\mu \in \coup(\mu_X,\mu_Y)}\lnorm \w_X(s,\cdot) - \w_Y(t,\cdot)\rnorm_{L^p(\mu)}, \qquad (s,t) \in X\times Y.
\tag{JE}
\label{eq:JE}
\end{align}
Note that $\w_X(s,\cdot)$ and $\w_Y(t,\cdot)$ are both measurable \cite[Proposition 2.34]{folland1999real}. Switching the arguments above produces the \emph{incoming joint eccentricity function}.
Taking the norm of an individual term gives the \emph{(outgoing) eccentricity function} $\eccout_{p,X}:X \r \R_+$:
\begin{align}
\eccout_{p,X}(s):= \lnorm \w_X(s,\cdot) \rnorm_{L^p(\mu_X)}, \qquad s \in X.
\tag{E}
\label{eq:E}
\end{align}
Flipping the arguments above produces the \emph{incoming eccentricity function}.
The norm of the preceding function is the \emph{$p$th $\size$ function} $\size_p : \Ngen \r \R_+$:
\begin{align}
\size_p(X):= \lnorm \eccout_{p,X} \rnorm_{L^p(\mu_X)} = \lnorm \w_X \rnorm_{L^p(\mu_X \otimes \mu_X)}. 
\tag{Sz}
\label{eq:Sz}
\end{align}

The size function is easily seen to be a network invariant: it compresses all the information in a network into a single real number. Theorem \ref{thm:stab-push} below shows that this compression occurs in a quantitatively stable manner. Notice that $\size_p$ can be computed exactly via a formula, and this computation is extremely cheap. Despite its simplicity, it can be very helpful as a first step in comparing networks. From a procedural perspective, given a network comparison task, one could compute $\size_p$ for different networks and compare these values to gain a coarse understanding of the discrepancies between the networks.

A priori, the connections between Equations (\ref{eq:E}) and (\ref{eq:JE}) to network invariants are somewhat unclear. We will use Theorem \ref{thm:stab-push} to clarify these connections in Remark \ref{rem:conn-lb-invar}, but we present the statements now for convenience. 
It will turn out that the invariant associated to Equation (\ref{eq:E}) is the map that takes a network $X$ to the distribution $( \eccout_{p,X})_*\mu_X$ over $\R$. The metric between distributions will be taken to be $W_p$, i.e. the codomain of this invariant is $(\prob(\R), W_p)$. 
Next, the invariant associated to Equation (\ref{eq:JE}) will turn out to be the map that takes $X$ to the distribution over distributions of $\w_X(x,\cdot)$. Specifically, it will be the pushforward of $\mu_X$ under the map $x \mapsto  \w_X(x,\cdot)_*\mu_X $. The codomain of this invariant will be $(\prob(\prob(\R)), W_p)$, where the ground metric on $\prob(\R)$ is also taken to be $W_p$. 

As a related construction, we note that given any $(X,\w_X,\mu_X)$, taking a pushforward of $\mu_X^{\otimes 2}$ via $\w_X$ yields a distribution over $\R$. This produces yet another invariant whose codomain is $(\prob(\R), W_p)$.

\begin{remark}[Local and global invariants]
Let $(X,\w_X,\mu_X \in \Ngen$. 
Both the $\size_p$ and $(\w_X)_*(\mu_X \otimes \mu_X)$ invariants are examples of \emph{global} invariants, in the sense that they incorporate data from the network without any reference to particular nodes in the network. 
In contrast, $(\eccout_{p,X})_*\mu_X$ and $(x \mapsto  \w_X(x,\cdot)_*\mu_X)_*\mu_X$ incorporate information at the level of individual nodes within the network, and constitute examples of \emph{local} invariants.
\end{remark}

We now state the main theorem of this section, which provides a hierarchy of lower bounds for $\dnp$.

\begin{theorem}[Hierarchy of lower bounds]
\label{thm:stab-push}
Let $(X,\w_X,\mu_X), (Y,\w_Y,\mu_Y) \in \Ngen$, and let $p\in [1,\infty]$. Let $C:X\times Y \r \R$ denote a cost matrix with entries 
$C(x,y) := W_p \lp \w_X(x,\cdot)_*\mu_X,   \w_Y(y,\cdot)_*\mu_Y \rp$. 
Then we have the following statements about Lipschitz stability, for $p \in [1,\infty]$:
\begin{align}
2\dnp(X,Y) &= \inf_{\mu \in \coup(\mu_X,\mu_Y)}\dis_p(\mu) = \inf_{\mu \in \coup(\mu_X,\mu_Y)} \lnorm \w_X - \w_Y \rnorm_{L^p(\mu^{\otimes 2})} \nonumber\\
&\geq  \inf_{\mu,\nu \in \coup(\mu_X,\mu_Y)} \lnorm \w_X - \w_Y \rnorm_{L^p(\nu \otimes \mu)} \nonumber\\
& \hspace{1.5in} = \inf_{\mu \in \coup(\mu_X,\mu_Y)} \lnorm \eccout_{p,X,Y} \rnorm_{L^p(\mu)}
\label{inv:TLB} \tag{TLB}\\ 
& \hspace{1.5in}=  \inf_{\mu \in \coup(\mu_X,\mu_Y)} 
\lnorm  C  \rnorm_{L^p(\mu)}
\label{inv:RTLB} \tag{$\R$-TLB}\\
& \geq \inf_{\mu \in \coup(\mu_X,\mu_Y)}\lnorm \eccout_{p,X} - \eccout_{p,Y} \rnorm_{L^p(\mu)} 
\label{inv:FLB}
\tag{FLB}\\ 
&  \hspace{1.5in}  =W_p\lp (\eccout_{p,X})_*\mu_X, (\eccout_{p,Y})_*\mu_Y  \rp.
\label{inv:RFLB}
\tag{$\R$-FLB}
\\
&\geq \lbar \size_p(X) - \size_p(Y) \rbar  \label{inv:SzLB} \tag{SzLB}\\
2\dnp(X,Y)  &\geq 
\inf \left \{\lnorm \w_X - \w_Y \rnorm_{L^p(\mu)} : \mu \in \coup(\mu_X^{\otimes 2}, \mu_Y^{\otimes 2}) \right \}  
\tag{SLB}
\label{inv:SLB} \\
&=
W_p \lp (\w_X)_*\mu_X^{\otimes 2}, (\w_Y)_*\mu_Y^{\otimes 2} \rp. 
\tag{$\R$-SLB}
\label{inv:RSLB}
\end{align}
Moreover, analogous bounds hold for the $\eccin$ variants as well. 
\end{theorem}

\begin{remark} The inequalities in Theorem \ref{thm:stab-push} appeared in the context of metric measure spaces as the First, Second, and Third Lower Bounds and their pushforwards in \cite{dgh-sm}. In the asymmetric context of the current paper, we obtain outgoing/incoming versions of the (\ref{inv:TLB}) and  (\ref{inv:FLB}) inequalities. The main development of the current paper is that we have equalities (\ref{inv:FLB})=(\ref{inv:RFLB}), (\ref{inv:SLB})=(\ref{inv:RSLB}), and (\ref{inv:TLB})=(\ref{inv:RTLB}). The equality (\ref{inv:TLB})=(\ref{inv:RTLB}) is especially important. A priori, each computation of $\eccout_{p,X,Y}(x,y)$ involves an OT problem that can be solved via linear programming methods. The equality (\ref{inv:TLB})=(\ref{inv:RTLB}) shows that this quantity is actually equal to the solution of an OT problem over the real line, which has a closed form solution. Finally we note that in the discrete case, all of the aforementioned equalities follow from \cite[Proposition 4.5]{schmitzer2013modelling}. The current theorem proves the equalities in the general setting. 

\end{remark}

\begin{remark}[Connecting lower bounds to network invariants]
\label{rem:conn-lb-invar}
The (\ref{inv:TLB}) lower bound arises by solving an OT problem with Equation (\ref{eq:JE}) as a cost matrix, and the (\ref{inv:FLB}) lower bound arises by solving an OT problem with a difference of terms described by Equation (\ref{eq:E}) as a cost matrix. By virtue of the equalities (\ref{inv:FLB})=(\ref{inv:RFLB}), (\ref{inv:SLB})=(\ref{inv:RSLB}), and (\ref{inv:TLB})=(\ref{inv:RTLB}), these lower bounds arise precisely as Lipschitz stability conditions on the network invariants described prior to the statement of Theorem \ref{thm:stab-push}.
\end{remark}

Before proving Theorem \ref{thm:stab-push}, we introduce some terminology from \cite[\S14A]{kechris1995classical}. A subset $A$ of a Polish space $X$ is \emph{analytic} if it is the continuous image of a Polish space $Y$. Equivalently, $A$ is analytic if there exists a Polish space $Y$ and a Borel subset $B\subseteq X\times Y$ such that $A = \pi_X(B)$, where $\pi_X$ is the canonical projection. Any Borel measurable map $f:X \r Y$, where $Y$ is Polish, maps analytic sets to analytic sets \cite[Proposition 14.4]{kechris1995classical}.

\begin{lemma}[Lemma 2.2, \cite{varadarajan1963groups}]
\label{lem:pushforward-analytic}
Let $X,Y$ be analytic subsets of Polish spaces equipped with the relative Borel $\s$-fields. Let $f:X \r Y$ be a surjective, Borel-measurable map. Then for any $\nu \in \prob(Y)$, there exists $\mu \in \prob(X)$ such that $\nu = f_*\mu$. 
\end{lemma}

The next lemma states that pushforwards of couplings are exactly the couplings between the pushforwards. This was shown in the special case of discrete spaces in \cite[Proposition 4.5]{schmitzer2013modelling}.

\begin{lemma}
\label{lem:distrib-invar}
Let $X,Y$ be Polish, and let $f:X \r \R$ and $g: Y \r \R$ be measurable. Let $T:X\times Y \r \R^2$ be the map $(x,y) \mapsto (f(x),g(y))$. 
Then we have:
\begin{align}
T_*\coup(\mu_X,\mu_Y) &= \coup(f_*\mu_X,g_*\mu_Y).
\label{eq:push-coupling}
\intertext{
Consequently, we have:}
W_p( f_*\mu_X, g_*\mu_Y) 
&= \inf_{\mu \in \coup(\mu_X,\mu_Y)} \lnorm f - g \rnorm_{L^p(\mu)}.
\label{eq:push-coupling-eq}
\end{align}
\end{lemma}

\begin{proof}[Proof of Lemma \ref{lem:distrib-invar}]
Let $\mu \in \coup(\mu_X,\mu_Y)$. 
It is standard \cite[7.1.6]{ambrosio2008gradient} that $T_*\mu \in \coup(f_*\mu_X,g_*\mu_Y)$, and hence $W_p( f_*\mu_X, g_*\mu_Y) \leq \lnorm f - g \rnorm_{L^p(\mu)}.$

For the ``$\supseteq$" containment of Equation (\ref{eq:push-coupling}), let $\nu \in \coup(f_*\mu_X, g_*\mu_Y)$. 
The map $T = (f,g)$ is measurable because $f$, $g$ are measurable. 
Next note that $X\times Y$ is Polish and hence analytic. 
Because $X\times Y$ is analytic and $T:X\times Y \r \R^2$ is a measurable map between Polish spaces, the image $T(X\times Y)$ is analytic \cite[Proposition 14.4]{kechris1995classical}.  
The map $T:X\times Y \r T(X\times Y)$ is surjective by construction. 
Then Lemma \ref{lem:pushforward-analytic} applies to the map $T:X\times Y \r T(X\times Y)$ and the restriction $\nu |_{T(X\times Y)} \in \prob(T(X\times Y))$.
Thus we obtain $\s \in \coup(\mu_X,\mu_Y)$ such that $T_*\s = \nu  |_{T(X\times Y)}$.
Finally note that $\nu$ is completely determined by its restriction to $T(X\times Y)$: for any $Z \in \borel(\R^2)$, we have $\nu(Z) = \nu(Z\cap T(X\times Y))$. 
Since $\nu  |_{T(X\times Y)}$ determines $\nu$, the existence of $\s$ such that $T_*\s = \nu  |_{T(X\times Y)}$ suffices to show the $\supseteq$ containment. The equality $W_p( f_*\mu_X, g_*\mu_Y) = \lnorm f - g \rnorm_{L^p(\mu)}$ follows immediately.

For Equation (\ref{eq:push-coupling-eq}), note that by a change of variables we have ($d_\R$ is just the standard distance on $\R$):
\[ \lnorm d_\R \rnorm_{L^p(T_*\mu)} = \lnorm f - g \rnorm_{L^p(\mu)}\]
Let $\nu \in \coup(f_*\mu_X, g_*\mu_Y)$. By the preceding work, $\nu = T_*\s$ for some $\s \in \coup(\mu_X,\mu_Y)$. Hence we have:
\[W_p( f_*\mu_X, g_*\mu_Y) = \inf_{\nu \in \coup(f_*\mu_X, g_*\mu_Y)} \lnorm d_\R \rnorm_{L^p(\nu)} = \inf_{\mu \in \coup(\mu_X,\mu_Y)} \lnorm f - g \rnorm_{L^p(\mu)}. \qedhere\]
\end{proof}

\begin{proof}[Proof of Theorem \ref{thm:stab-push}]
Inequality (\ref{inv:TLB}) holds because $\nu$ is allowed to vary and thus we infimize over a larger set. Next fix $(x,y) \in X\times Y$. Applying Lemma \ref{lem:distrib-invar} Equation (\ref{eq:push-coupling-eq}), we have 
\[ C(x,y) = W_p\lp \w_X(x,\cdot)_*\mu_X, \w_Y(y,\cdot)_*\mu_Y \rp = \inf_{\nu \in \coup(\mu_X,\mu_Y)} \lnorm \w_X(x,\cdot) - \w_Y(y,\cdot) \rnorm_{L^p(\nu)} = \eccout_{p,X,Y}(x,y).\]
This proves (\ref{inv:TLB})=(\ref{inv:RTLB}). Next, for any $\nu \in \coup(\mu_X,\mu_Y)$, we have by Minkowski's inequality: 
\[\lnorm \w_X(x,\cdot) - \w_Y(y,\cdot) \rnorm_{L^p(\nu)} \geq \lbar \lnorm \w_X(x,\cdot) \rnorm_{L^p(\nu)} - \lnorm \w_Y(y,\cdot) \rnorm_{L^p(\nu)} \rbar = \lbar \eccout_{p,X}(x) - \eccout_{p,Y}(y) \rbar\]
This shows (\ref{inv:TLB})$\geq$(\ref{inv:FLB}). The equality (\ref{inv:FLB})=(\ref{inv:RFLB}) follows by another application of Lemma \ref{lem:distrib-invar}. Next, for any $\mu \in \coup(\mu_X,\mu_Y)$, another application of Minkowski's inequality yields:
\[ \lnorm \eccout_{p,X} - \eccout_{p,Y} \rnorm_{L^p(\mu)} \geq 
\lbar \lnorm \eccout_{p,X} \rnorm_{L^p(\mu)} - \lnorm \eccout_{p,Y} \rnorm_{L^p(\mu)} \rbar
 = \lbar \size_p(X) - \size_p(Y) \rbar.\]
This shows (\ref{inv:FLB})$\geq$(\ref{inv:SzLB}).

Finally for (\ref{inv:SLB}), let $\mu$ denote the minimizer of $\dis_p$ (invoking Theorem \ref{thm:opt-coup}) and define $\s:= \mu \otimes \mu$. 
Then $\s \in \coup(\mu_X^{\otimes 2}, \mu_Y^{\otimes 2})$. Hence
\[ 2\dnp(X,Y) = \lnorm \w_X - \w_Y \rnorm_{L^p(\mu\otimes \mu)} 
=  \lnorm \w_X - \w_Y \rnorm_{L^p(\s)} \geq \inf_{\nu \in \coup(\mu_X^{\otimes 2}, \mu_Y^{\otimes 2})} \lnorm \w_X - \w_Y \rnorm_{L^p(\nu)} .\]
This shows (\ref{inv:SLB}). The equality (\ref{inv:SLB})=(\ref{inv:RSLB}) follows by applying Lemma \ref{lem:distrib-invar} Equation (\ref{eq:push-coupling-eq}) with $f = \w_X$ and $g = \w_Y$. \qedhere
     
\end{proof}

\subsection{Interleaving stable invariants}
\label{sec:int-stable}
We now present a novel family of invariants that satisfies a different type of stability. Let $(X,\w_X,\mu_X) \in \Nm$, and let $p \in [1,\infty]$. For each $t \in \R$ and $x \in X$, define the quantity 
\[ \eccout_{p,X}(x,t) := \lnorm \w_X(x,\cdot) \mathbf{1}_{ \{\w_X(x,\cdot) \leq t\}} \rnorm_{L^p(\mu_X)}.\]
This is an overload of notation, but the meaning should be clear from the presence of the second parameter. Note that $\{\w_X(x,\cdot) \leq t \}$ is measurable, and so $\mathbf{1}_{ \{\w_X(x,\cdot) \leq t\}}$ is measurable. Hence the integral is well-defined.

\begin{remark} For a metric space $(X,d_X,\mu_X)$, the quantity $\lnorm \mathbf{1}_{ \{d_X(x,\cdot) \leq t\}} \rnorm^p_{L^p(\mu_X)}$ is just the measure of the ball of radius $t$ centered at $x$. 

\end{remark}

Next, the $p$th \emph{sublevel} size function is defined for each $(X,\w_X,\mu_X) \in \Nm$ and $t \in \R$ by writing
\begin{align}
\subw_{p,t}(X) = \lnorm \eccout_{p,X}(\cdot,t) \rnorm_{L^p(\mu_X)}. 
\tag{subSz}
\label{eq:subSz}
\end{align}
This function is a network invariant. Note that by the Fubini-Tonelli theorem, we can also write $\subw_{p,t}(X) = \lnorm \w_X \mathbf{1}_{ \{\w_X \leq t \} } \rnorm_{L^p(\mu_X \otimes \mu_X)}$. Both formulations are used below.
\begin{example}
In \cite[Example 5.7]{dghlp}, it was shown that the \emph{$1$-diameter} invariant (referred to as $\size_1$ in this paper) does not discriminate between spheres of different dimensions. Specifically, it was shown that
\[ \size_1(\sph^n) =  \frac{\pi}{2} \tag*{for any $n \in \N.$}\]
 We now show via explicit computations that the map $t \mapsto \subw_{1,t}$ does distinguish between spheres.
For each $n \in \N$, let $\sph^n$ denote the $n$-sphere with the geodesic metric and normalized volume measure. 
For each $n \in \N$, let $S_n$ denote the surface area of $\sph^n$. We have:
\[S_1 = 2\pi, \, S_2 = 4\pi,\, S_3 = 2\pi^2, \, S_4 = \frac{8}{3}\pi^2.\]
The following formula gives $\subw$ for $\sph^n$, $n \in \N$.

\begin{proposition} 
\label{prop:p-subdiam}
Fix $p \in [1,\infty)$. Let $n \in \N$, $n \geq 2$, and $0 \leq t \leq \pi$. Then,

\[\subw_{p,t}(\sph^n)^p = \frac{S_{n-1}}{S_n}\int_0^t \ph^p \sin^{n-1}(\ph)\, d\ph.  \]

For $n = 1$, we have:

\[\subw_{p,t}(\sph^1)^p = \frac{t^{p+1}}{(p+1)\pi}. \]

\end{proposition}

By applying this result, we obtain $\subw_{1,t}(\sph^1) = \frac{t^2}{2\pi}$ and $\subw_{1,t}(\sph^2) = \frac{\sin(t) - t\cos(t)}{2}$, where $0 \leq t \leq \pi$. 
Plots of these functions are provided in Figure \ref{fig:s1-s2}. Note that by having access to the functions, instead of just the function values at $t = \pi$ (which corresponds to the prior $\size_1$ result of \cite[Example 5.7]{dghlp}), we are able to distinguish between spheres of different dimensions. 

An interesting consequence of the preceding result, along with the result that $\size_1(\sph^n) =  \frac{\pi}{2}$ for all $n \in \N$, is the following identity for $n \geq 2$:
\begin{align}
\frac{S_{n-1}}{S_n} \int_{0}^\pi \ph \sin^{n-1}(\ph) \,d\ph = \frac{\pi}{2}.
\end{align}
In particular, this identity and the formula in Proposition \ref{prop:p-subdiam} explain why the 1-diameter (i.e. $\size_1$) cannot distinguish between spheres, and why $\subw_{1,t}$ is able to do so.

\begin{proof}[Proof of Proposition \ref{prop:p-subdiam}]

Let $n = 1$. We obtain the formula as a line integral. Let $r(\theta) = (\cos \theta, \sin \theta)$ be a parametrization of the circle, where $\theta \in [0,2\pi)$. Using symmetry, we have the following for $0 \leq t \leq \pi$: 

\begin{align*}
\subw_{p,t}(\sph^1)^p = \frac{2}{S_1} \int_0^t \theta^p \| r'(\theta)\| \, d\theta =  \frac{2}{S_1} \left. \frac{\theta^{p+1}}{p+1}\right\rvert_0^t = \frac{t^{p+1}}{(p+1)\pi}.
\end{align*}

Next let $n \geq 2$. In hyperspherical coordinates, the area element of $\sph^n$ is given by 
\[ \sin^{n-1}(\ph_1) \sin^{n-2}(\ph_2)\cdots \sin(\ph_{n-1}) \, d\ph_1 d\ph_2 \cdots d\ph_{n},\]
where the limits of integration are $[0,\pi]$ for $\ph_1,\ldots, \ph_{n-1}$, and $[0,2\pi]$ for $\ph_n$.
As an example, we have:
\[ \subw_{p,t}(\sph^2)^p = \frac{1}{S_2} \int_0^{2\pi} \int_0^t \ph_1^p\sin \ph_1 \, d\ph_1 d\ph_2.\]
Generalizing to larger values of $n$, we have:
\begin{align*}
 \subw_{p,t}(\sph^n)^p &= \frac{1}{S_n} \int_{\ph_n = 0}^{2\pi} \int_{\ph_{n-1} = 0}^\pi \cdots \int_{\ph_1 = 0}^t \ph_1^p \, \sin^{n-1}(\ph_1) \sin^{n-2}(\ph_2) \cdots \sin (\ph_{n-1}) \, d\ph_1 d\ph_2 \cdots d\ph_n\\
&= \frac{S_{n-1}}{S_n} \int_{\ph_1 = 0}^t \ph_1^p \, \sin^{n-1}(\ph_1) \, d\ph_1. \qedhere
\end{align*}
\end{proof}

\end{example}

Having motivated $\subw$ by at least a theoretical application, we now proceed to its stability.

\begin{theorem}[Interleaving stability of $\subw$]
\label{thm:subsz-stab}
Let $p\in [1,\infty]$, $t\in \R$ and let $(X,\w_X,\mu_X)$, $(Y,\w_Y,\mu_Y) \in \Nm$. Define $\e:= \dgpzero(X,Y)$. Then we have the following interleaving stability:
\begin{align*}
\subw_{p,t}(X) & \leq \e + \subw_{p,t+\e}(Y),\\
\subw_{p,t}(Y) & \leq \e + \subw_{p,t+\e}(X).
\end{align*}
\end{theorem}

\begin{proof} 
We show the first statement.
Invoking Lemma \ref{lem:gp-gw}, we write $\e=\dninf(X,Y)$. Using Theorem \ref{thm:opt-coup}, let $\mu\in \coup(\mu_X,\mu_Y)$ be an optimal coupling for which $\dninf(X,Y)=\e$ is achieved.
Let $B:= \{(x,y,x',y') \in (X\times Y)^2 : \lbar \w_X(x,x') - \w_Y(y,y') \rbar \geq \e\}$. 
Let $G$ denote the complement of $B$, i.e. $G :=  \{(x,y,x',y') \in (X\times Y)^2 : \lbar \w_X(x,x') - \w_Y(y,y') \rbar < \e\}$. 
By the definition of $\e$, we have $\mu^{\otimes 2}(B) = 0$, and hence $\mu^{\otimes 2}(G) = 1$. 
Also define $H:= G \cap \lp \{\w_X \leq t\} \times Y^2 \rp$. Then we have:
\begin{align}
\subw_{p,t}(X) &= \lnorm \eccout_{p,X}(\cdot,t) \rnorm_{L^p(\mu_X)}
= \lnorm  \w_X \mathbf{1}_{ \{\w_X \leq t \} } \rnorm_{L^p(\mu_X^{\otimes 2})}
= \lnorm  \w_X \mathbf{1}_{ \{\w_X \leq t \} \times Y^2 } \rnorm_{L^p(\mu^{\otimes 2})}
\nonumber\\
&= \lnorm  \w_X \mathbf{1}_{H} \rnorm_{L^p(\mu^{\otimes 2})}
= \lnorm  \lp \w_X  - \w_Y + \w_Y \rp \mathbf{1}_{H} \rnorm_{L^p(\mu^{\otimes 2})}
\nonumber \\
&\leq \lnorm  \lp \w_X  - \w_Y \rp \mathbf{1}_H \rnorm_{L^p(\mu^{\otimes 2})} + \lnorm \w_Y \mathbf{1}_{H} \rnorm_{L^p(\mu^{\otimes 2})}
\nonumber \\
& < \e + \lnorm \w_Y \mathbf{1}_{ \{ \w_Y \leq t+\e \}} \rnorm_{L^p(\mu_Y^{\otimes 2})}
= \e + \subw_{p,t+\e}(Y).
\label{eq:subsz-1}
\end{align}
Here the third equality holds because $\mu$ is a coupling measure, and the fourth equality holds because $\mu^{\otimes 2}(G) = 1$. The first inequality holds by Minkowski's inequality. The first part of the second inequality holds because $|\w_X(x,x') - \w_Y(y,y)| < \e$ on $H$, and the second part holds because $\w_Y(y,y') \leq \w_X(x,x') + \e \leq t+\e$ on $H$. Finally note that repeating the argument with the roles of $X$ and $Y$ switched completes the proof.
\end{proof}

\begin{remark} While not applied in the current paper, we may also consider a \emph{superlevel} size function $\supw_{p,t}(X):= \lnorm \w_X \mathbf{1}_{\{\w_X \geq t \}} \rnorm_{L^p(\mu_X\otimes \mu_X)}$. In the setup of Theorem \ref{thm:subsz-stab}, this invariant satisfies the following interleaving stability:
\begin{align*}
\supw_{p,t}(X) &\leq \e + \supw_{p,t-\e}(Y) \\
\supw_{p,t}(Y) &\leq \e + \supw_{p,t-\e}(X). 
\end{align*}
To see this, note that the proof of Theorem \ref{thm:subsz-stab} carries through until the  step in Inequality (\ref{eq:subsz-1}). In this case, for any $(x,y,x',y') \in H$ we have $\w_Y(y,y') > \w_X(x,x') - \e \geq t - \e$, thus $\mathbf{1}_H$ reduces to $\mathbf{1}_{ \{ \w_Y \geq t- \e\}}$.
\end{remark}

\subsubsection{Lower bounds for spheres}
\label{sec:sphere-lbounds}

Fix $n, m \in \N$. We now invoke Theorem \ref{thm:subsz-stab} to obtain lower bounds on $\dgpzero(\sph^n,\sph^m)$. The explicit value of $\dgpzero(\sph^n,\sph^m)$ is unknown in the existing literature, even for $n=1,\, m=2$. 

Consider the family $\mf{F}:= \{ f:[0,\pi] \r \R_+ : f \text { increasing} \}.$ For each $f \in \mf{F}$ and $\e \in [0,\pi]$, define $f^\e$ by writing, for each $t \in [0,\pi]$,
\[f^\e(t) := 
\begin{cases}
f(t+\e) + \e &: t+\e \in [0,\pi]\\
f(\pi) + \e &: \text{ otherwise.}
\end{cases}
\]
Next define the \emph{interleaving distance} $\di$ on $\mf{F}$ by writing, for each $f,g \in \mf{F}$, 
\[ \di(f,g) := \inf\{ \e \geq 0 : f \leq g^\e \text{ and } g \leq f^\e\}.\]
This $\di$ is a pseudometric on $\mf{F}$. 
Next, for $p \in [1,\infty)$, Define $f_p,g_p:[0,\pi] \r \R$ by writing:
\[ f_p(t):= \subw_{p,t}(\sph^n), \qquad g_p(t):= \subw_{p,t}(\sph^m) \tag*{for all $t \in [0,\pi].$}\]
Define $\eta:= \dgpzero(\sph^n, \sph^m)$. Applying Theorem \ref{thm:subsz-stab}, we have $f_p \leq g_p^\eta$ and $g_p \leq f_p^\eta$. Thus $\dgpzero(\sph^n,\sph^m) \geq \di(f_p,g_p)$. Moreover, by the triangle inequality of $\di$, we have
\[\dgpzero(\sph^n,\sph^m) \geq \di(f_p,g_p) \geq \lbar \di(f_p,h) - \di(h,g_p) \rbar,\]
for arbitrary $h \in \mf{F}$. In particular, setting $h \equiv 0$, we have $\di(f_p,h) = \subw_{p,\pi}(\sph^n) = \size_p(\sph^n)$ and $\di(g_p,h) = \subw_{p,\pi}(\sph^m) = \size_p(\sph^m)$. Thus we obtain a $\size_p$ bound:
\[\dgpzero(\sph^n, \sph^m) \geq \lbar \size_p(\sph^n) - \size_p(\sph^m) \rbar.\]
This bound can be easily improved using different choices of $h \in \mf{F}$. 

Using the explicit formula of Proposition \ref{prop:p-subdiam}, we are able to computationally obtain lower bounds on $\dgpzero(\sph^n, \sph^m)$. Set $p=1$, $n=1$, and $m=2$. Then $f_1(t) = \frac{t^2}{2\pi}$ and $g_1(t) = \frac{\sin(t) - t\cos(t)}{2}$. Plots of $f_1$ and $g_1$ are shown in Figure \ref{fig:s1-s2}. 

Through Matlab simulations, we find $\dgpzero(\sph^1, \sph^2) = \dninf(\sph^1,\sph^2) \geq \di(f_1,g_1) \geq \mathbf{0.17}$. To contrast this with a previously known lower bound, we refer to \cite[Remark 5.16]{dghlp}, where the lower bound $\dnt(\sph^1,\sph^2) \geq 0.0503$ was obtained. Because $\dninf \geq \dnt$, this previously known lower bound yields $\dninf(\sph^1,\sph^2) \geq 0.0503$. Our new lower bound of $0.17$ improves this threefold.

\begin{figure}
\centering
\begin{minipage}[t]{0.35\textwidth}
\includegraphics[width=\textwidth]{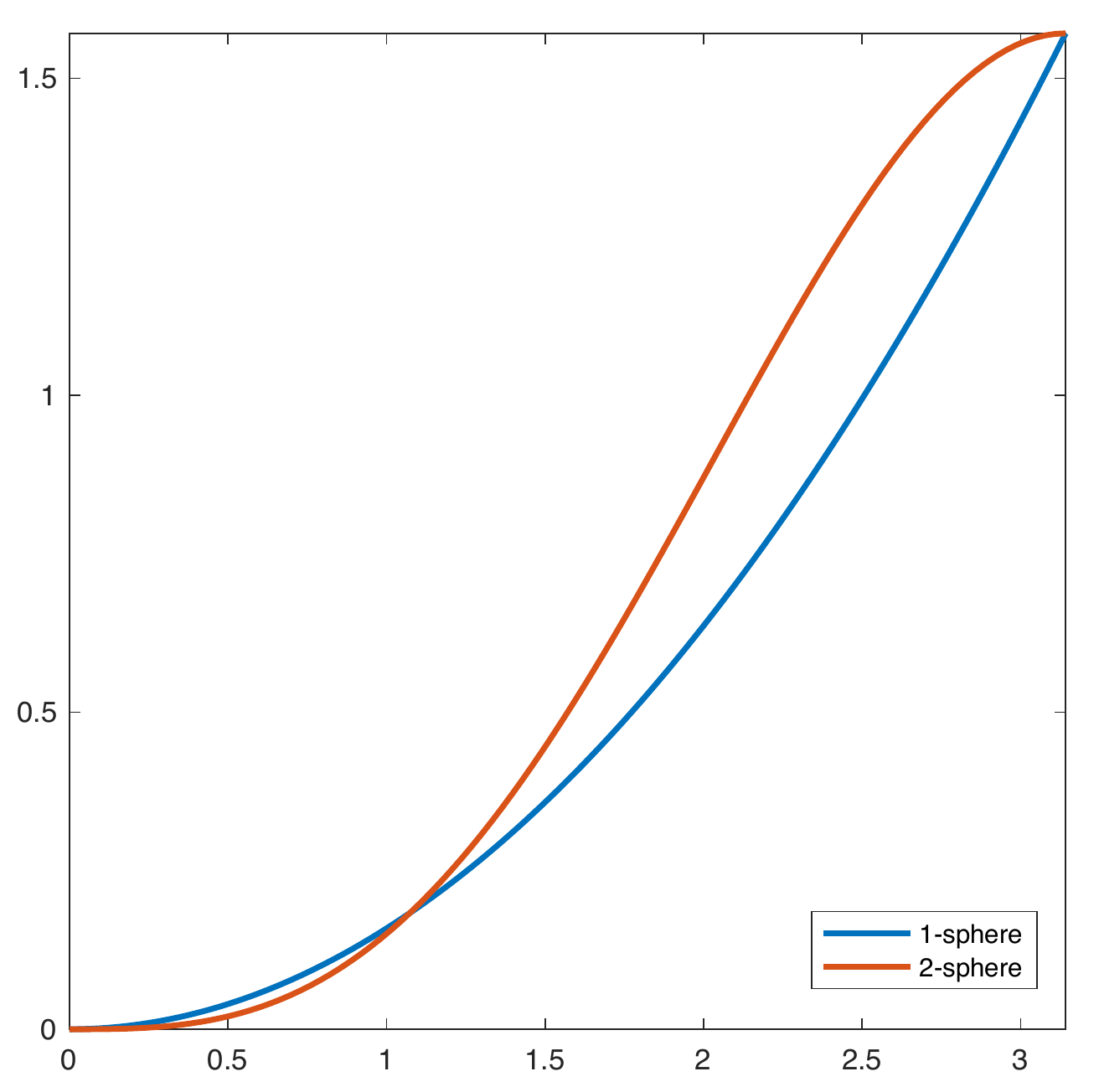}
\end{minipage} 
\begin{minipage}[t]{0.35\textwidth}
\includegraphics[width=\textwidth]{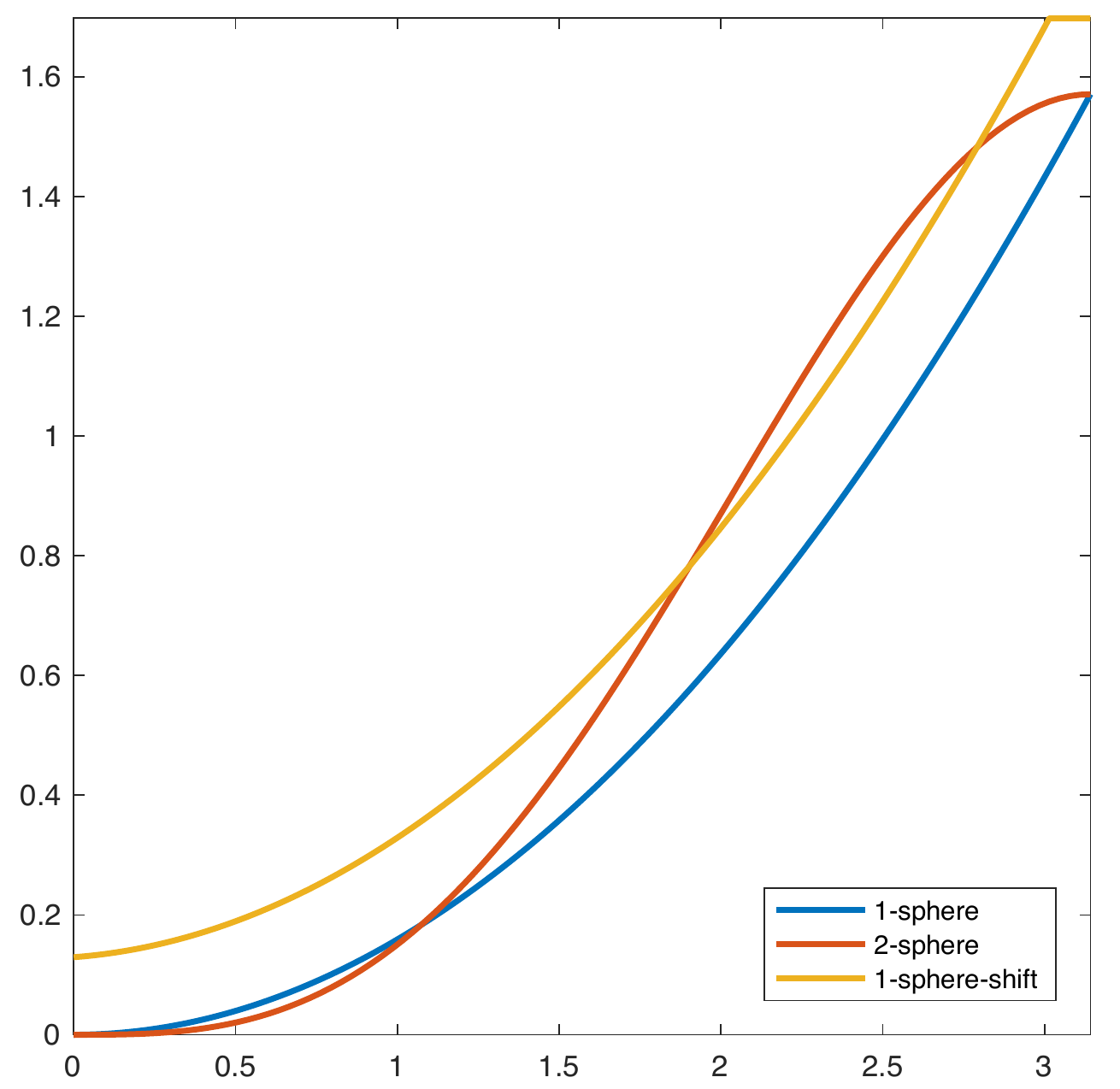}
\end{minipage}
\caption{\textbf{Left:} Plots of $f_1$ and $g_1$ as described in \S\ref{sec:sphere-lbounds}. \textbf{Right:} Plots of $f_1,\, g_1,$ and a shifted version of $f_1$.}
\label{fig:s1-s2}
\end{figure}

\section{Experiments}
\label{sec:exp}

\subsection{Computational aspects}

Numerical experiments in \cite{dgh-sm, hendrikson} involved using an alternate optimization procedure to estimate a local minimum of the GW objective. The methods in \cite{s16, pcs16} used an entropically regularized GW objective (ERGW) which led to fast algorithms. These methods remain valid in the setting of (possibly asymmetric) networks. To complement the existing literature, in this section we present the use of the (\ref{inv:TLB}) lower bound to compute dissimilarities between asymmetric networks. By virtue of the equality (\ref{inv:TLB})=(\ref{inv:RTLB}), this lower bound can be computed by solving a single general OT problem over a cost matrix obtained by solving OT problems over the real line. This is practical because OT problems over $\R$ have closed form solutions, with the caveat that computing all these OT problems is still the main bottleneck in computations. In comparable demonstrations, the ERGW of \cite{pcs16} is orders of magnitude faster, but a standard warning about ERGW is that it is prone to numerical infeasibility issues (see Appendix \ref{sec:computation}). For networks of several hundred nodes, the  (\ref{inv:RTLB}) can be computed exactly at reasonable speed, i.e in less than a minute in Matlab on a 2.3 GHz Intel i5 CPU with 8 GB memory. Our experiments show that (\ref{inv:RTLB}) works well in discriminating networks. 

Next we review the formula for computing OT over $\R$ (see \cite[Remark 2.19]{villani2003topics}) . Let networks $(X,\w_X,\mu_X)$, $(Y,\w_Y,\mu_Y)$ and measurable functions $f:X \r \R$, $g:Y \r \R$ be given. In the $\eccout$ setting, $f = \w_X(x,\cdot)_*\mu_X$ and $g = \w_Y(y,\cdot)_*\mu_Y$. Then let $F, G: \R \r [0,1]$ denote the cumulative distribution functions of $f$ and $g$:
\begin{align*}
F(t) := \mu_X( \{x \in X: f(x) \leq t\}), \qquad G(t):= \mu_Y( \{y \in Y : g(y) \leq t \}).
\end{align*}
The generalized inverses $F\inv:[0,1] \r \R$, $G\inv:[0,1] \r \R$ are given as:
\begin{align*}
F\inv(t) := \inf\{ u \in \R : F(u) \geq t\}, \qquad G\inv(t) := \inf\{ u \in \R : G(u) \geq t\}.
\end{align*}
Then for $p \geq 1$, one has:
\begin{align}
\inf_{\mu \in \coup(f_*\mu_X, g_*\mu_Y)} \int_{\R \times \R} |a - b|^p \diff\mu(a,b) = \int_0^1 | F\inv(t) - G\inv(t) |^p \diff t. \label{eq:OT-R}
\end{align}
For $p = 1$, one obtains a reformulation that incurs lower computational cost, at least in a naive implementation:
\begin{align}
\inf_{\mu \in \coup(f_*\mu_X, g_*\mu_Y)} \int_{\R \times \R} |a - b| \diff\mu(a,b) = \int_{\R} | F(t) - G(t) | \diff t. \label{eq:OT-R-p1}
\end{align}

In our experiments, we computed both the $\eccout$ and $\eccin$ versions of (\ref{inv:RTLB}) and take their maximum as the lower bound. All computations were done for $p=2$. For Wasserstein distance computations, we used the \texttt{mexEMD} code accompanying \cite{pcs16}. Our code and data are available on \url{https://github.com/samirchowdhury/GWnets}.

In a prior version of this paper, before the equality (\ref{inv:TLB})=(\ref{inv:RTLB}) was proved in full generality, we were faced with the problem of solving an ensemble of OT problems over the space $X\times Y$. At the time, we resorted to using entropic regularization to compute the (\ref{inv:TLB}) in a reasonable amount of time. A priori this could also have been done by directly solving the associated linear programs, using e.g. \texttt{mexEMD}. 
While entropic regularization is not used in the current paper, we briefly report on these prior approaches in Appendix \ref{sec:computation}.

\subsection{The network stochastic block model}
\label{sec:sbm-net}

We now describe a generative model for random networks, based on the popular stochastic block model for sampling random graphs \cite{abbe2017community}. The current network SBM model we describe is a composition of Gaussian distributions. However, the construction can be adjusted easily to work with other distributions.

Fix a number of \emph{communities} $N \in \N$. For $1\leq i,j \leq N$, fix a mean $\mu_{ij}$ and a variance $\s^2_{ij}$. This collection $\mc{G}:= \{\mc{N}(\mu_{ij},\s^2_{ij}) : 1\leq i,j\leq N \}$ of $N^2$ independent Gaussian distributions comprise the network SBM. 

To sample a random network $(X,\w_X)$ of $n$ nodes from this SBM, start by fixing $n_i \in \N, 1\leq i \leq N$ such that $\sum_i n_i = n$. For $1\leq i \leq N$, let $X_i$ be a set with $n_i$ points. Define $X:= \cup_{i=1}^n X_i$. Next sample each node weight as $\w_X(x,x') \sim \mc{N}(\mu_{ij},\s^2_{ij})$, where $x \in X_i$ and $x' \in X_j$. Finally, the pair $(X,\w_X)$ is equipped with the uniform measure $\mu_X$ that assigns a mass of $1/n$ to each point.

We now describe the specifics of two experiments on clustering a collection of network SBMs.

\subsection{Experiment: SBMs from cycle networks.}
\label{sec:exp-sbm-community}

Let $N \in \N$, and let $v=[v_1,\ldots, v_N]$ be an $N\times 1$ vector. Define the right-shift operator $\rho$ by $\rho([v_1,\ldots,v_N]) = [v_N,v_1,\ldots, v_{N-1}]$. The \emph{cycle network} $G_N(v)$ is defined to be the $N$-node network whose weight matrix is given by $[v^T, \rho(v)^T, (\rho^2(v))^T,\ldots, (\rho^{N-1}(v))^T]$. The cycle network definition appears elsewhere in the literature, see e.g. \cite{pph}. An illustration is provided in Figure \ref{fig:cyc6}.

\begin{figure}
\begin{center}
\begin{tikzpicture}[every node/.style={font=\footnotesize}]
\begin{scope}[draw]
\node[circle,draw](1) at (-1,1.5){$x_1$};
\node[circle,draw](2) at (1,1.5){$x_2$};
\node[circle,draw](3) at (1.75,0){$x_3$};
\node[circle,draw](4) at (1,-1.5){$x_4$};
\node[circle,draw](5) at (-1,-1.5){$x_5$};
\node[circle,draw](6) at (-1.75,0){$x_6$};
\end{scope}

\begin{scope}[draw]
\path[->] (1) edge [bend left=20] node[above,pos=0.5]{$1$} (2);
\path[->] (2) edge [bend left=20] node[above right,pos=0.5]{$1$} (3);
\path[->] (3) edge [bend left=20] node[below right,pos=0.5]{$1$} (4);
\path[->] (4) edge [bend left=20] node[above,pos=0.5]{$1$} (5);
\path[->] (5) edge [bend left=20] node[below left,pos=0.5]{$1$} (6);
\path[->] (6) edge [bend left=20] node[above left,pos=0.5]{$1$} (1);
\end{scope}

\begin{scope}[xshift=6cm]
\draw[step=0.5cm,color=gray] (-1.5,-1.5) grid (1.5,1.5);

\node at (-1.25,+1.75) {$x_1$};
\node at (-0.75,+1.75) {$x_2$};
\node at (-0.25,+1.75) {$x_3$};
\node at (+0.25,+1.75) {$x_4$};
\node at (+0.75,+1.75) {$x_5$};
\node at (+1.25,+1.75) {$x_6$};

\node at (-1.75,+1.25) {$x_1$};
\node at (-1.75,+0.75) {$x_2$};
\node at (-1.75,+0.25) {$x_3$};
\node at (-1.75,-0.25) {$x_4$};
\node at (-1.75,-0.75) {$x_5$};
\node at (-1.75,-1.25) {$x_6$};

\node at (-1.25,+1.25) {$0$};
\node at (-0.75,+1.25) {$1$};
\node at (-0.25,+1.25) {$2$};
\node at (+0.25,+1.25) {$3$};
\node at (+0.75,+1.25) {$4$};
\node at (+1.25,+1.25) {$5$};

\node at (-1.25,+0.75) {$5$};
\node at (-0.75,+0.75) {$0$};
\node at (-0.25,+0.75) {$1$};
\node at (+0.25,+0.75) {$2$};
\node at (+0.75,+0.75) {$3$};
\node at (+1.25,+0.75) {$4$};

\node at (-1.25,+0.25) {$4$};
\node at (-0.75,+0.25) {$5$};
\node at (-0.25,+0.25) {$0$};
\node at (+0.25,+0.25) {$1$};
\node at (+0.75,+0.25) {$2$};
\node at (+1.25,+0.25) {$3$};

\node at (-1.25,-0.25) {$3$};
\node at (-0.75,-0.25) {$4$};
\node at (-0.25,-0.25) {$5$};
\node at (+0.25,-0.25) {$0$};
\node at (+0.75,-0.25) {$1$};
\node at (+1.25,-0.25) {$2$};

\node at (-1.25,-0.75) {$2$};
\node at (-0.75,-0.75) {$3$};
\node at (-0.25,-0.75) {$4$};
\node at (+0.25,-0.75) {$5$};
\node at (+0.75,-0.75) {$0$};
\node at (+1.25,-0.75) {$1$};

\node at (-1.25,-1.25) {$1$};
\node at (-0.75,-1.25) {$2$};
\node at (-0.25,-1.25) {$3$};
\node at (+0.25,-1.25) {$4$};
\node at (+0.75,-1.25) {$5$};
\node at (+1.25,-1.25) {$0$};
\end{scope}

\end{tikzpicture}
\end{center}
\caption{A cycle network on 6 nodes corresponding to the weight matrix obtained by right-shifting the vector $[0,1,2,3,4,5]^T$. Note that the weights are highly asymmetric.}
\label{fig:cyc6}
\end{figure}
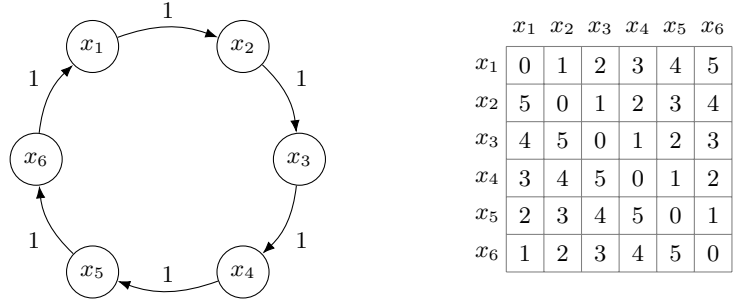

In our first experiment on network SBMs, we started with an $N\times 1$ vector of means $v$ and used this to generate $G_N(v)$. We then used $G_N(v)$ as the matrix of means. To keep the experiment simple, we fixed the matrix of variances to be the $N\times N$ matrix whose entries are all $5$s. We made 5 choices of $v$, and sampled 10 networks for each choice. The objective was then to see how well the (\ref{inv:RTLB}) could split the collection of 50 networks into 5 classes corresponding to the 5 different community structures. The different parameters used in our experiments are listed in Table \ref{tab:sbm-means}.

\begin{table}
\centering
\begin{minipage}{0.45\linewidth}
\tabcolsep=0.2cm
\begin{tabular}{|c||c|c|c|}
\hline 
Class \# & $N$ & $v$ & $n_i$  \\
\hline 
1 & 5 & [0,25,50,75,100] & 10 \\ 
\hline 
2 & 5 & [0,50,100,150, 200]  & 10 \\ 
\hline 
3 & 5 & [0,25,50,75,100]  & 20 \\ 
\hline 
4 & 2 & [0,100] & 25 \\ 
\hline
5 & 5 & [-100,-50,0,50,100] & 10\\ 
\hline 
\end{tabular} 
\end{minipage}
\begin{minipage}{0.2\linewidth}
\tabcolsep=0.15cm
\begin{tabular}{c}
Sample cycle network of means
\end{tabular}
\begin{tabular}{|c|c|c|c|c|}
\hline 
0 & 25 & 50 & 75 & 100 \\ 
\hline 
100 & 0 & 25 & 50 & 75 \\ 
\hline 
75 & 100 & 0 & 25 & 50\\ 
\hline 
50 & 75 & 100 & 0 & 25 \\ 
\hline 
25 & 50 & 75 & 100 & 0\\ 
\hline 
\end{tabular} 
\end{minipage}
\medskip
\caption{\textbf{Left:} The five classes of SBM networks corresponding to the experiment in \S\ref{sec:exp-sbm-community}. $N$ refers to the number of communities, $v$ refers to the vector that was used to compute a table of means via $G_5(v)$, and $n_i$ is the number of nodes in each community. \textbf{Right:} $G_5(v)$ for $v = [0,25,50,75,100]$.} 
\label{tab:sbm-means}
\end{table}

Class 1 is our reference; compared to this reference, class 2 differs in its edge weights, class 3 differs in the number of nodes in each community, class 4 differs in the number of communities, and class 5 differs by having a larger proportion of negative edge weights. The (\ref{inv:RTLB}) results in Figure \ref{fig:tlb-sbm-dissim} show that classes 1 and 3 are treated as being very similar, whereas the other classes are all mutually well-separated. This is consistent, because $\dn$ is not sensitive to the size of the networks (cf. Theorem \ref{thm:dnp-pseudo}). One interesting suggestion arising from this experiment is that the (\ref{inv:RTLB}) can be used for network simplification: given a family of networks which are all at low (\ref{inv:RTLB}) distance to each other, it may be reasonable to retain only the smallest network in the family as the ``minimal representative" network.

\begin{figure}
\centering
\begin{minipage}[t]{0.45\textwidth}
\includegraphics[width=\textwidth]{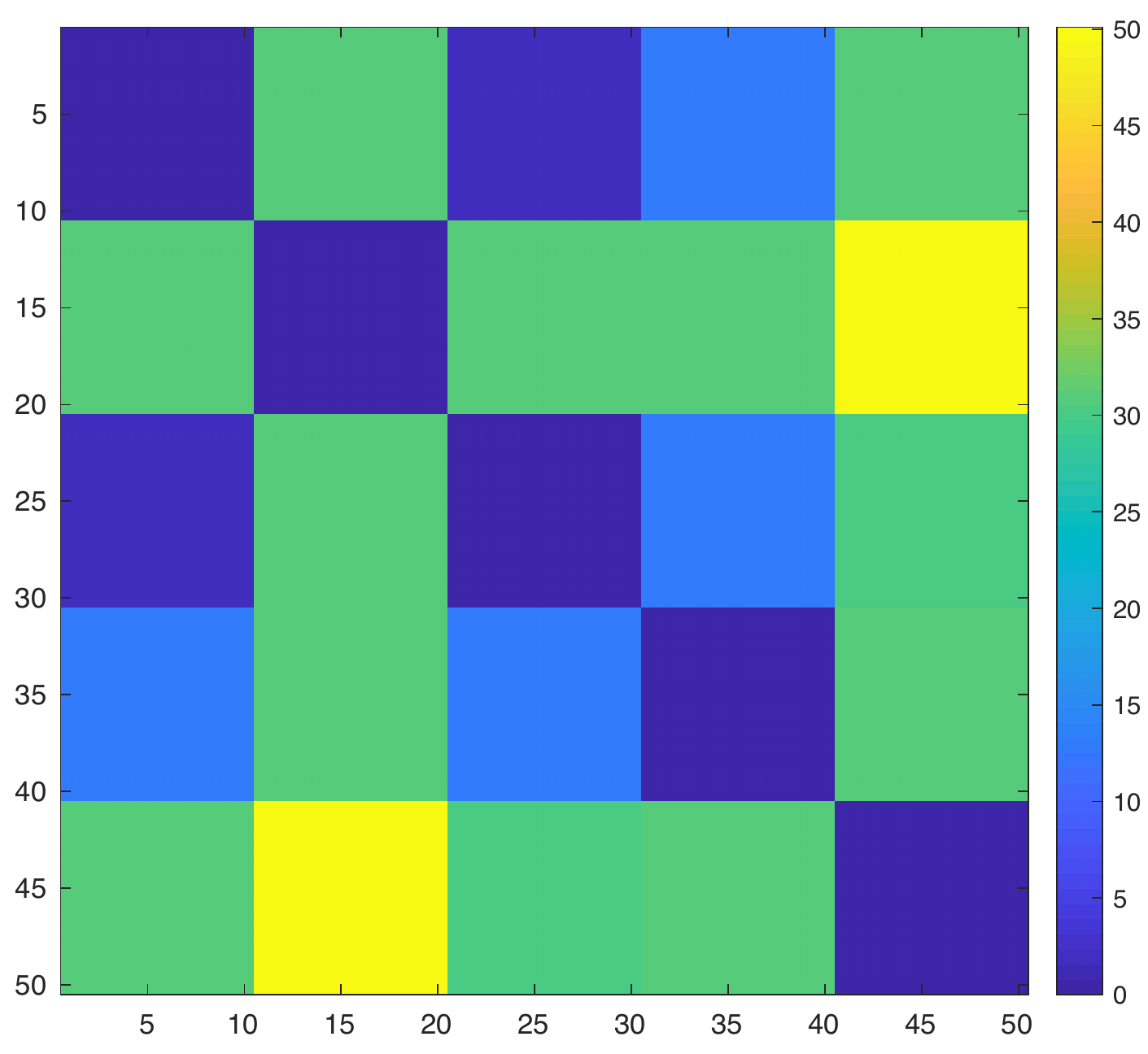}
\end{minipage}
\begin{minipage}[t]{0.45\textwidth}
\includegraphics[width=\textwidth]{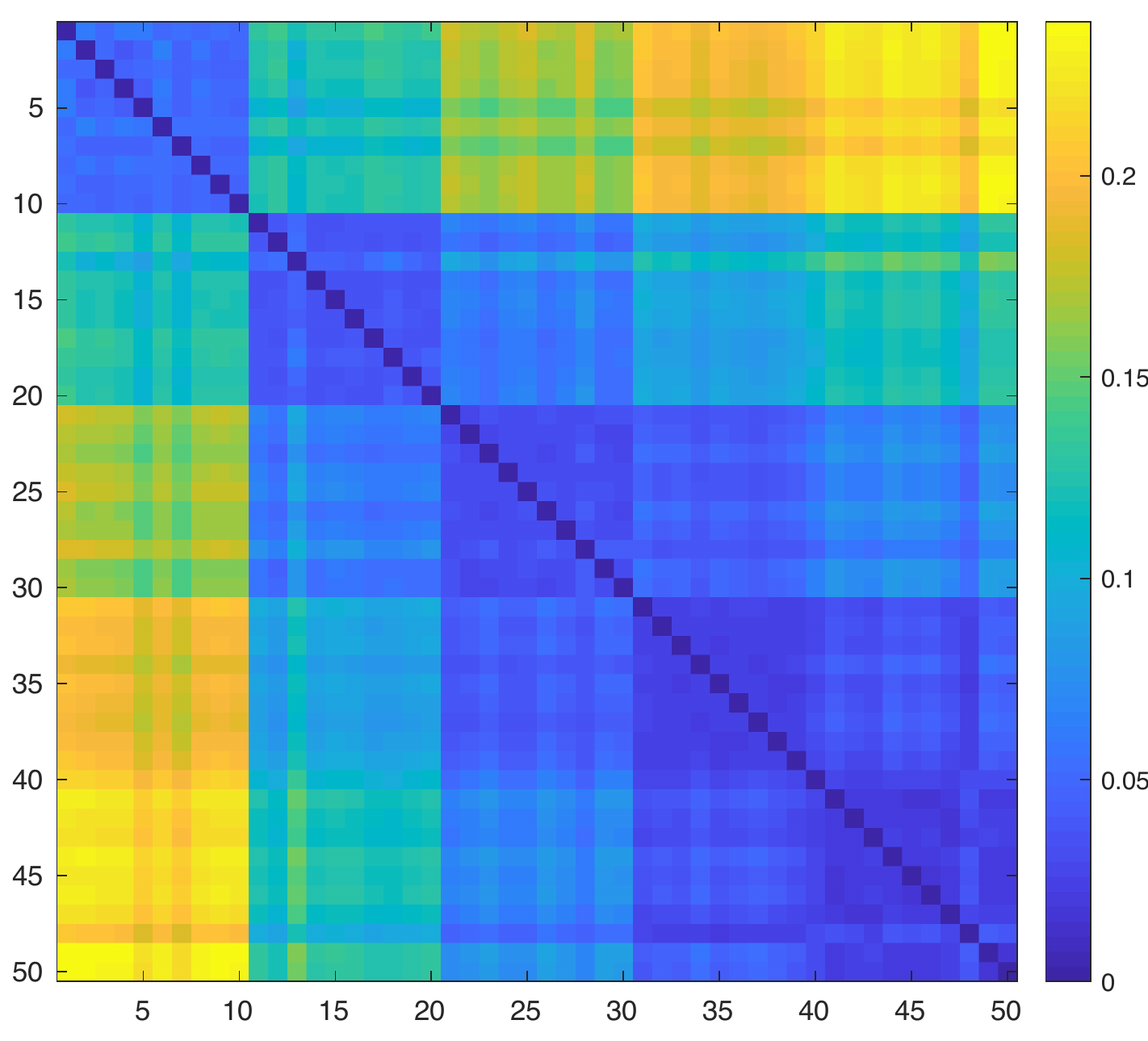}
\end{minipage}
\caption{\textbf{Left:} TLB dissimilarity matrix for SBM community networks in \S\ref{sec:exp-sbm-community}. Classes 1 and 3 are similar, even though networks in Class 3 have twice as many nodes as those in Class 1. Classes 2 and 5 are most dissimilar because of the large difference in their edge weights. Class 4 has a different number of communities than the others, and is dissimilar to Classes 1 and 3 even though all their edge weights are in comparable ranges. \textbf{Right:} (\ref{inv:RTLB}) dissimilarity matrix for two-community SBM networks in \S\ref{sec:exp-sbm-two}.}
\label{fig:tlb-sbm-dissim}
\end{figure}

\subsection{Experiment: Two-community SBMs with sliding means}
\label{sec:exp-sbm-two}

\begin{table}
\centering
\begin{tabular}{|c||c|c|c|}
\hline 
Class \# & $N$ & $v$ & $n_i$  \\
\hline 
1 & 2 & [0,0] & 10 \\ 
\hline 
2 & 2 & [0,5]  & 10 \\ 
\hline 
3 & 2 & [0,10]  & 10 \\ 
\hline 
4 & 2 & [0,15] & 10 \\ 
\hline
5 & 2 & [0,20] & 10\\ 
\hline 
\end{tabular} 
\medskip
\caption{Two-community SBM networks as described in \S\ref{sec:exp-sbm-two}.}
\label{tab:sbm-two}
\end{table}

Having understood the interaction of the (\ref{inv:RTLB}) with network community structure, we next investigated how the (\ref{inv:RTLB}) behaves with respect to edge weights. In our second experiment, we used a $2\times 1$ means vector $v$, and varied $v$ as $[0,0],[0,5],\ldots,[0,20]$ (see Table \ref{tab:sbm-two}). The SBM means were then given by $G_2(v)$ for the various choices of $v$. The variances were fixed to be the all 5s matrix. The edge weight histograms of the resulting SBM networks then looked like  samples from two Gaussian distributions, with one of the Gaussians sliding away from the other. Finally, we normalized each network by its largest weight in absolute value, so that its normalized edge weights were in $[-1,1]$. 

The purpose of this experiment was to test the performance of (\ref{inv:RTLB}) on SBMs coming from a mixture of Gaussians. Note that normalization ensures that simpler invariants such as the $\size$ invariant would likely fail in this setting. The (\ref{inv:RTLB}) still performs reasonably well in this setting, as illustrated by the dissimilarity matrix in Figure \ref{fig:tlb-sbm-dissim}. The linear color gradient is consistent with the ``sliding means" network structure. 

\subsection{Experiment: Real migration networks}
\label{sec:exp-migration}

For an experiment involving real-world networks, we compared global bilateral migration networks produced by the World Bank \cite{worldbank, ozden2011earth}. The data consists of 10 networks, each having 225 nodes corresponding to countries/administrative regions. The $(i,j)$-th entry in each network is the number of people living in region $i$ who were born in region $j$. The 10 networks comprise such data for male and female populations in 1960, 1970, 1980, 1990, and 2000. When extracting the data, we removed the entries corresponding to refugee populations, the Channel Islands, the Isle of Man, Serbia, Montenegro, and Kosovo, because the data corresponding to these regions was incomplete/inconsistent across the database. We assigned uniform mass to the nodes.

The result of applying the (\ref{inv:RTLB}) to this dataset is illustrated in Figure \ref{fig:tlb-migration}. To better understand the dissimilarity matrix, we also computed its single linkage dendrogram. The dendrogram suggests that between 1960 and 1970, both male and female populations had quite similar migration patterns. Within these years, however, migration patterns were more closely tied to gender. This effect is also seen between 1980 and 1990, although male migration in 1990 is more divergent. Finally, migration rates are similar for both male and female populations in 2000, and they are different from migration patterns in prior years. 

The labels in the dissimilarity matrix are as follows: 1-5 correspond to ``f-1960'' through ``f-2000'', and 6-10 correspond to ``m-1960'' through ``m-2000''. The color gradient in the dissimilarity matrix suggests that within each gender, migration patterns change in a way that is parametrized by time. This reflects the shifts in global technological and economical forces which make migration attractive and/or necessary with time.

\begin{figure}
\centering
\begin{minipage}[t]{0.45\textwidth}
\includegraphics[width=\textwidth]{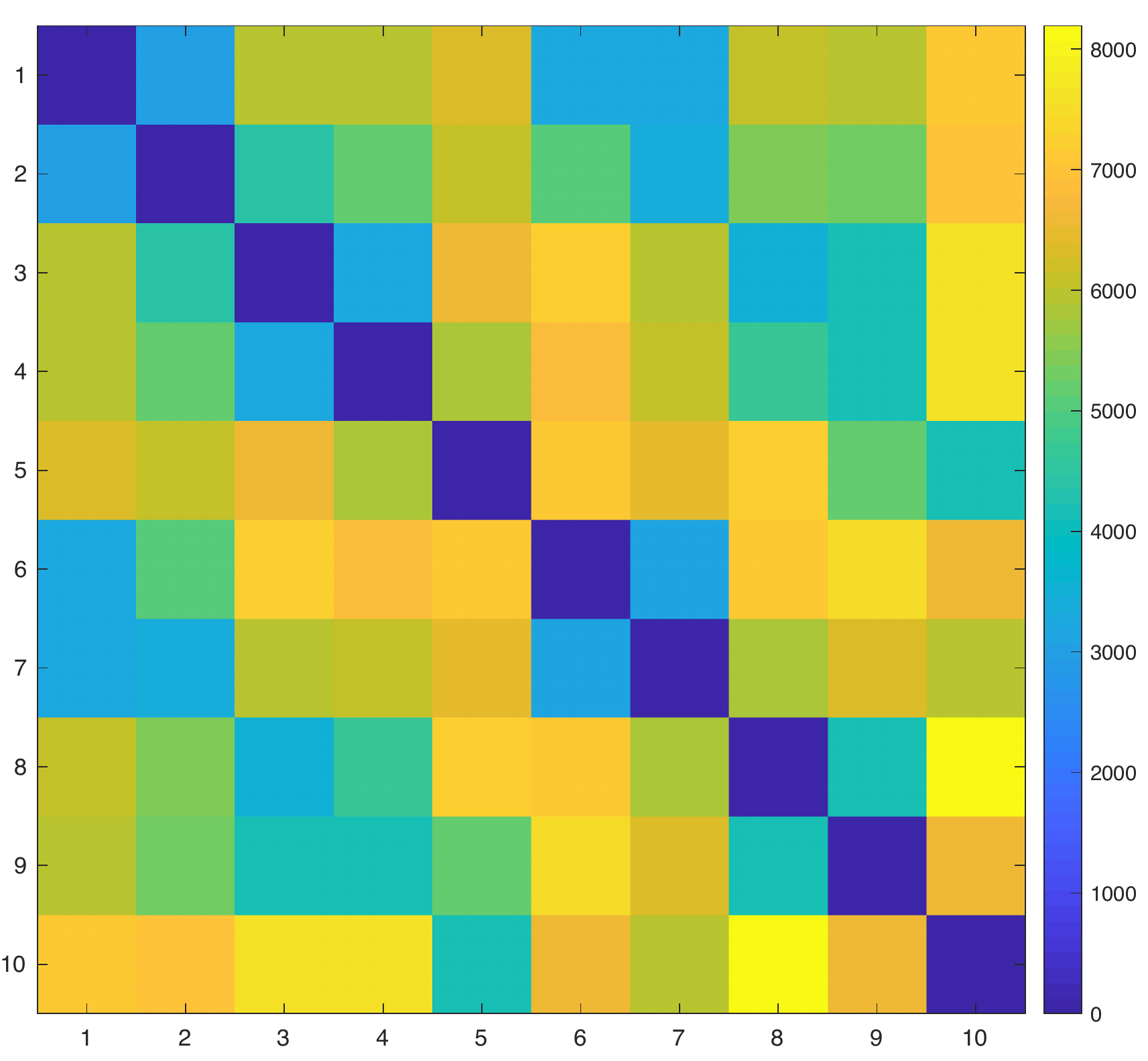}
\end{minipage}
\begin{minipage}[t]{0.45\textwidth}
\includegraphics[width=\textwidth]{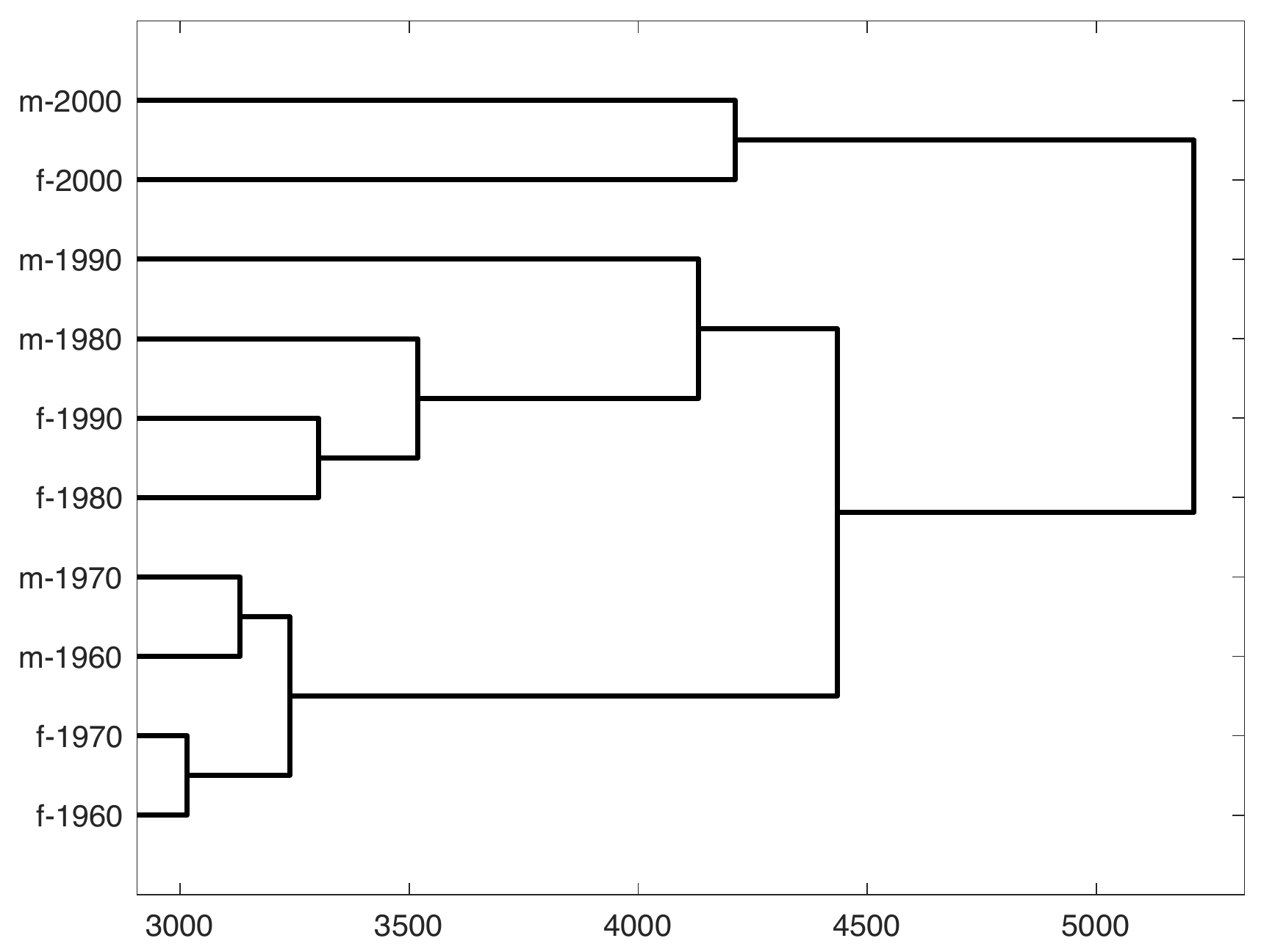}
\end{minipage}
\caption{Result of applying the (\ref{inv:RTLB}) to the migration networks in \S\ref{sec:exp-migration}. \textbf{Left:} Dissimilarity matrix. Nodes 1-5 correspond to female migration from 1960-2000, and nodes 6-10 correspond to male migration from 1960-2000. \textbf{Right:} Single linkage dendrogram. Notice that overall migration patterns change in time, but within a time period, migration patterns are grouped according to gender.}
\label{fig:tlb-migration}
\end{figure}

\section{Discussion} \label{sec:discussion}
We have presented the GW distance as a valid pseudometric on the space of all directed, weighted networks. The crux of this approach is that even though the GW distance was originally formulated for metric measure spaces, the structure of the GW distance automatically forces a metric structure on networks. This yields the insight that the metric structure on the ``space of spaces" is not inherited from the metric on the ground spaces. In particular, while there are several metrics on networks that are combinatorial in nature and hence hard to compute/sensitive to outliers, this GW metric is considerably more relaxed.  The OT-based network invariants that we have presented yield lower bounds on the GW distance which at most involve linear programming, and hence are readily computable. 
Finally, we tested our methods on a range of network datasets. The SBM network model that we defined for these tests will likely yield useful benchmarks for such network methods in future applications.

\medskip

\noindent
\textbf{Acknowledgments} This project was supported by NSF grants IIS-1422400, DMS-1723003, and TRIPODS-1740761. We are especially grateful to the anonymous reviewers for their detailed feedback and comments, and also to Justin Solomon for useful insights regarding computation.



\bibliographystyle{alpha}
\bibliography{biblio}

\newpage
\appendix 
\section{Computation via entropic regularization}
\label{sec:computation}

Entropic regularization (ER), as used in \cite{cuturi2013sinkhorn} and further developed in \cite{benamou2015iterative, s16, pcs16, vayer2019optimal}, can be used in an iterative algorithm that approximates a local minimum of the GW objective \cite{s16,pcs16}. In this section, we describe some heuristics that we found useful when applying ER-based techniques on network data. The main issue that we deal with is the following: initializing an ER-objective for networks having very different edge weights may create cost matrices with values below machine precision, which causes computations to blow up. A related issue that we found was the problem of ``entropic bias", which can be dealt with using well-understood techniques \cite{feydy2019}.

We first explain the notion of entropic regularization and associated difficulties with numerical stability. Throughout this section, we write $M$ to denote a cost matrix depending on the edge weights of networks $(X,\w_X,\mu_X)$, $(Y,\w_Y,\mu_Y)$. This $M$ could be the GW objective, as in \cite{pcs16}, or one of the lower bound matrices from Theorem \ref{thm:stab-push}.

\subsection{Numerical stability of entropic regularization}
\label{sec:entropy-geometry}

Let $(X,\w_X,\mu_X)$, $(Y,\w_Y,\mu_Y)$ be networks with $|X| = m$, $|Y| = n$. For a general $m\times n$ cost matrix $M$, one may consider the entropically regularized optimal transport problem below, where $\lambda \geq 0 $ is a regularization parameter and $H$ denotes entropy:
\begin{align*}
\inf_{p \in \coup(\mu_X,\mu_Y)}\sum_{i,j}M_{ij}p_{ij} - \frac{1}{\lambda}H(p), \qquad
H(m) = - \sum_{i,j}p_{ij}\log p_{ij}.
\end{align*}

As shown in \cite{cuturi2013sinkhorn}, the solution to this problem has the form $\diag(a)*K*\diag(b)$, where  $K:= e^{ - \lambda M}$ is a kernel matrix and $a,b$ are nonnegative scaling vectors in $\R^m, \R^n$, respectively. Here $*$ denotes matrix multiplication, and exponentiation is performed elementwise. An approximation to this solution can be obtained by iteratively scaling $K$ to have row and column sums equal to $\mu_X$ and $\mu_Y$, respectively, and iterating until convergence. This is described in Algorithm \ref{alg:sinkhorn}.

\begin{algorithm}
\caption{Sinkhorn algorithm \cite{cuturi2013sinkhorn}}
\label{alg:sinkhorn}
\begin{algorithmic}
\Procedure{sinkhorn}{$M,\lambda, mA,mB$} \Comment{$M$ an $m\times n$ cost matrix, $mA,mB$ prob. measures}
	\State $a \gets 1_m$, $b \gets 1_n$ \Comment{scaling updates, initialize as all-ones vectors}
	\State $K_{ij} \gets {\exp}(-\lambda M_{ij})$ \Comment{initialize kernel}
	\Repeat
		\State $b\gets mB./(K'a)$,  $a \gets mA./(Kb)$
	\Until{convergence}
	\State \Return{$\diag(a)K\diag(b)$}
\EndProcedure
\end{algorithmic}
\end{algorithm}

As pointed out in \cite{schmitzer2019stabilized, chizat2018scaling, chizat-thesis}, using a large value of $\lambda$ (corresponding to a small regularization) leads to numerical instability, where values of $K$ can go below machine precision and entries of the scaling udpates $a,b$ can also blow up. For example, Matlab will interpret $e^{-1000}$ as 0, which is a problem with even a moderate choice of $\lambda = 200$ and $M_{ij} = 50$. Theoretically, it is necessary to have $K$ be a positive matrix for the Sinkhorn algorithm to converge to the correct output \cite{sinkhorn1964relationship, sinkhorn1967diagonal}. Practitioners use a range of techniques to deal with the numerical instability, e.g. occasionally ``absorbing" extreme values of $a,b$ into the kernel $K$ (log-domain absorption), or gradually updating $\lambda$ after starting with a conservative value (see \cite{chizat-thesis} for more details). 

Specifically in the network setting, initializing the kernel matrix $K$ can be tricky due to the wide range of edge weights in the dataset: both within a network and between different networks. For example, in the migration network database, the migration into a large country like the USA is separated by several orders of magnitude from that of a smaller country, such as Austria. Furthermore, migration values differ significantly between years, e.g. between 1960 and 2000.

As discussed in \cite{chizat-thesis}, many entries of the stabilized kernel obtained as above could be below machine precision, but the entries corresponding to those on which the optimal plan is supported are likely to be above the machine limit. Indeed, this sparsity may even be leveraged for additional computational tricks. 

The techniques for stabilizing the entropy regularized OT problem are not the focus of our work, but because these considerations naturally arose in our computational experiments, we describe some strategies we undertook that are complementary to the techniques available in the current literature. In order to provide a perspective complementary to that presented in \cite{chizat-thesis}, we impose the requirement that \emph{all} entries of the kernel matrix remain above machine precision.  

\medskip
\noindent
\textbf{Initializing in the log domain.} A simple adaptation of the ``log domain absorption" step referred to above yields a ``log initialization" method that works well in most cases for initializing $K$ to have values above machine precision. To explain this method, we first present an algorithm (Algorithm \ref{alg:sinkhorn-log}) for the log domain absorption method. We follow the presentation provided in \cite{chizat-thesis}, making notational changes as necessary.

\begin{algorithm}
\caption{Sinkhorn with partial log domain steps}
\label{alg:sinkhorn-log}
\begin{algorithmic}
\Procedure{sinkhornLog}{$M,\lambda, mA,mB$} \Comment{$M$ an $m\times n$ cost matrix, $mA,mB$ prob. measures}
	\State $a \gets 1_m$, $b \gets 1_n$ \Comment{scaling updates}
	\State $u \gets 0_m$, $v \gets 0_n$ \Comment{log domain storage of large $a,b$}
	\State $K_{ij} \gets {\exp}(\lambda(-M_{ij} + u_i + v_j))$ \Comment{initialize kernel}
	\While{stopping criterion not met} 
		\State $b\gets mB./(K'a)$
		\State $a \gets mA./(Kb)$
		\If{$\max(\max(a),\max(b)) > \texttt{threshold}$}
			\State $u \gets u + (1/\lambda)\log(a)$ \Comment{store $a,b$ in $u,v$}
			\State $v \gets v + (1/\lambda)\log(b)$
			\State $K_{ij} \gets {\exp}(\lambda(-M_{ij} + u_i + v_j))$ \Comment{absorb $a,b$ into $K$}
			\State $a \gets 1_m$, $b \gets 1_n$ \Comment{after absorption, reset $a,b$}
		\EndIf
	\EndWhile
	\State \Return{$\diag(a)K\diag(b)$}
\EndProcedure
\end{algorithmic}
\end{algorithm}

Notice that in Algorithm \ref{alg:sinkhorn-log}, $K$ might already have values below machine precision at initialization. To circumvent this, we can add a preprocessing step that yields a stable initialization of $K$. This is outlined in Algorithm \ref{alg:log-initial}. An important point to note about Algorithm \ref{alg:log-initial} is that the user needs to choose a function \texttt{decideParam($\a,\b$)} which returns a ``translation factor" $\g$, where $\a$ and $\b$ are as stated in the algorithm. This number $\g$ should be such that $\exp(-\lambda\b + 2\lambda\g)$ is above machine precision, but $\exp(-\lambda\a + 2\lambda\g)$ is not too large. The crux of Algorithm \ref{alg:log-initial} is that by choosing large initial scaling vectors $a,b$ and immediately absorbing them into the log domain, the extreme values of $M$ are canceled out before exponentiation.

\begin{algorithm}
\caption{Log domain initialization of $K$}
\label{alg:log-initial}
\begin{algorithmic}
\Procedure{logInitialize}{$M,\lambda$} \Comment{$M$ an $m\times n$ cost matrix}
	\State $\alpha \gets \min(M)$, $\beta \gets \max(M)$ \Comment{scan $M$ for max and min values}
	\State $\gamma \gets \texttt{decideParam($\a,\b$)}$ \Comment{\texttt{decideParam} is an independent 	function}
	\State $a \gets {\exp}(-\lambda\gamma) 1_m$, $b \gets \texttt{exp}(-\lambda\gamma)  1_n$
	\State $u \gets  0_m$, $v \gets 0_n$
	\State $K_{ij} \gets {\exp}(\lambda(-M_{ij} + \gamma+ \gamma))$ \Comment{$K$ is stably initialized}
	\State perform rest of \Call{sinkhornLog}{} as usual	
\EndProcedure
\end{algorithmic}
\end{algorithm}

\medskip
\noindent
\textbf{A geometric interpretation in the $p=2$ case.} The preceding initialization method has its limitations: depending on how far $\min(M),\max(M)$ are spread apart, the log initialization step might not be able to yield an initial kernel $K$ that has all entries above machine precision and below the machine limit. In such a case, one recourse is to choose a different value of $\lambda$. Thus when given a database of networks $X_1,\ldots, X_n$ and cost matrices arising from comparing these networks, one may need to choose $\lambda_{ij} = \lambda_{ji}$ for each pair $\{X_i,X_j\}$. It turns out that these potentially different $\lambda$ values can be related to a global $\lambda^*$ value by rescaling the networks in a geometric manner, using observations from \cite{sturm2012space}. This is described below. In what follows, we always have $p=2$.

\subsubsection{Sturm's cosine rule construction}
Let $(X,\w_X,\mu_X), (Y,\w_Y,\mu_Y) \in \Nm$. 
Recall from Example \ref{ex:1-pt-dnp} that $\dnt(X,N_1(0)) = \tfrac{1}{2}\size_2(X)$. 
Define $s:=\tfrac{1}{2}\size_2(X,\w_X,\mu_X)$, $t:=\tfrac{1}{2}\size_2(Y,\w_Y,\mu_Y)$. 
For an optimal coupling $\mu \in \coup(\mu_X,\mu_Y)$, we have:
\begin{align}
\dnt(X,Y)^2 &= \frac{1}{4} \int\int  \w_X(x,x')^2 +  \w_Y(y,y')^2 - 2\w_X(x,x')\w_Y(y,y') \diff\mu(x,y)\diff\mu(x',y') \nonumber \\
&= s^2 + t^2 - \frac{1}{2}\int\int  \w_X(x,x') \w_Y(y,y')\diff\mu(x,y)\diff\mu(x',y'), \label{eq:dn-quadratic-1}
\end{align}
where the first equality holds because $|a-b|^2 = \la a-b, a-b \ra = |a|^2 + |b|^2 - 2ab$ for all $a,b\in \R$, and the last equality holds because $\w_X(x,x'), \w_Y(y,y')$ do not depend on $\mu_Y,\mu_X$, respectively.
Sturm \cite[Lemma 4.2]{sturm2012space} observed the following ``cosine rule" structure. Define 
\begin{align}
\w_X' := \frac{\w_X}{2s},
\qquad
\w_Y' := \frac{\w_Y}{2t}. \label{eq:scaling-reference}
\end{align}
Then $\size_2(X,\w'_X) = \frac{1}{2s}\size_2(X,\w_X) = 1 = \frac{1}{2t}\size_2(Y,\w_Y) = \size_2(Y,\w'_Y)$. 
A geometric fact about this construction is that $(X,\w_X,\mu_X)$, $(Y,\w_Y,\mu_Y)$ lie on geodesic rays connecting $\mc{X}:= (X,\w_X',\mu_X)$ and $\mc{Y}:= (Y,\w_Y',\mu_Y)$ respectively to $N_1(0)$. Actually, once $(X,\w_X,\mu_X)$ and $(Y,\w_Y,\mu_Y)$ are chosen, the geodesic rays are automatically defined to be given by the scalar multiples of $\w_X,\w_Y$. Then we independently define $\mc{X}$ and $\mc{Y}$ to be representatives of the weak isomorphism class of networks at $\dnt$ distance $1/2$ from $N_1(0)$ that lie on these geodesics. We illustrate a related situation in Figure \ref{fig:scaling}, and refer the reader to \cite{sturm2012space} for further details. Implicitly using this geometric fact, we fix $\mc{X}, \mc{Y}$ as above and treat $(X,\w_X,\mu_X)$, $(Y,\w_Y,\mu_Y)$ as $2s$ and $2t$-scalings of $\mc{X}$ and $\mc{Y}$, respectively (i.e. such that Equation \ref{eq:scaling-reference} is satisfied). Then we have:
\begin{align}
&4\dnt((X,\w_X,\mu_Y),(Y,\w_X,\mu_Y))^2 - 4s^2 - 4t^2 \nonumber \\
&= \int\int 4s^2 \w'_X(x,x')^2 + 4t^2 \w'_Y(y,y')^2 - 8st \w'_X(x,x')\w'_Y(y,y') \diff\mu(x,y)\diff\mu(x',y') - 4s^2 - 4t^2 \nonumber \\
&= 4s^2 \size_2(X,\w'_X)^2+ 4t^2 \size_2(Y,\w'_Y)^2 - 4s^2 - 4t^2 - 8st \int\int  \w'_X(x,x')\w'_Y(y,y') \diff\mu(x,y)\diff\mu(x',y') \nonumber \\
&= -8st\int\int\w_X'(x,x')\w_Y'(y,y')\diff\mu(x,y)\diff\mu(x',y'), \label{eq:dn-quadratic-2}
\end{align}
where the last equality holds because $\size_2(X,\w'_X) = 1 = \size_2(Y,\w'_Y)$. Since $(X,\w_X,\mu_X), (Y,\w_Y,\mu_Y)$ were $2s,2t$-scalings of $\mc{X}$ and $\mc{Y}$ for arbitrary $s,t > 0$, this shows in particular that the quantity 
\[(1/2st)\lp \dnt((X,\w_X,\mu_Y),(Y,\w_Y,\mu_Y))^2 - s^2 -t^2\rp \tag{cosine rule}\] 
depends only on the reference networks $\mc{X}$ and $\mc{Y}$, and is independent of $s$ and $t$. 

\subsubsection{Interpretation of $\lambda$ and rescaling}
\label{sec:lambda-rescaling}

Suppose now that we are in a setting where $\lambda^* > 0$, $(X,\w_X,\mu_X)$, $(Y,\w_Y,\mu_Y)$, and a cost matrix $M$ depending on $\w_X,\w_Y$ are all fixed. Suppose also that $e^{-\lambda^*M}$ contains values below machine precision, and $\lambda_{XY} > 0$ is such that $e^{-\lambda_{XY}M}$ has all values above machine precision. Then one may define $M^* := \frac{\lambda_{XY}}{\lambda^*}M$, so that $e^{-\lambda^*M^*} = e^{-\lambda_{XY}M}$. Here $M^*$ is a rescaled cost matrix, and in typical use cases, it is the cost matrix obtained from rescaled weights $\w_X', \w_Y'$. For example, if $M = \w_X\w_Y$ (as in the integrand of Equation \ref{eq:dn-quadratic-1}, also see \cite{pcs16}), then $M^* = \w_X'\w_Y'$, where $\w_X'$, $\w_Y'$ are rescaled from $\w_X$, $\w_Y$ by $\sqrt{\lambda_{XY}/\lambda^*}$. By the observations from \cite{sturm2012space} presented above, these rescalings are compatible with the geometry of $(\Nm,\dnt)$, in the sense that the rescaled networks lie on geodesics connecting the original networks to the basepoint $N_1(0)$. This is illustrated in Figure \ref{fig:scaling}. See \cite{sturm2012space} for more details about the geodesic structure of gauged measure spaces; the analogous results hold for $(\Nm,\dnt)$.

\begin{figure}[h]
\centering
\def\svgwidth{0.4\textwidth}
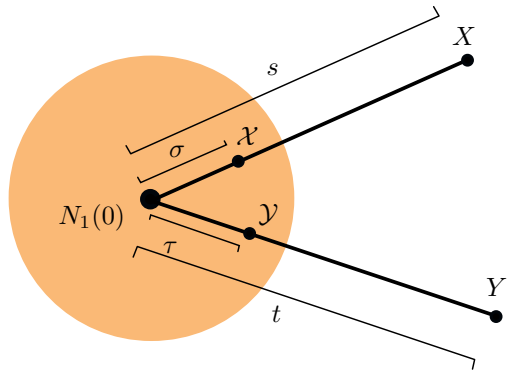
\caption{Interaction between the entropic regularization parameter and rescalings, cf. \S\ref{sec:lambda-rescaling}. Choosing a regularization parameter $\lambda_{XY}$ depending on the edge weights of $X$ and $Y$ is essentially the same as using a fixed parameter $\lambda^*$ with (edge-weight) rescaled versions $\mc{X}, \mc{Y}$ of $X$ and $Y$. Here $\mc{X}$, $\mc{Y}$ live on geodesic rays connecting $N_1(0)$ to $X$ and $Y$. The letters $s,\s,t,\t$ represent $\dnt$-distances.}
\label{fig:scaling}
\end{figure}

\end{document}